\documentclass{article}

\usepackage{PRIMEarxiv}

\usepackage[utf8]{inputenc} 
\usepackage[T1]{fontenc}    
\usepackage{hyperref}       
\usepackage{url}            
\usepackage{booktabs}       
\usepackage{amsmath,amssymb,amsfonts,amsthm}       
\usepackage{nicefrac}       
\usepackage{microtype}      
\usepackage{lipsum}
\usepackage{fancyhdr}       
\usepackage{graphicx}       
\graphicspath{{media/}}     

\theoremstyle{plain}
\newtheorem{theorem}{Theorem}
\newtheorem{lemma}[theorem]{Lemma}
\newtheorem{example}[theorem]{Example}
\theoremstyle{definition}
\newtheorem{definition}[theorem]{Definition}

\title{Matrix product and quasi-twisted codes in one class
}

\author{
  Ramy Taki Eldin \\
  Faculty of Engineering, Ain Shams University, Cairo, Egypt\\
  Egypt University of Informatics, Knowledge City, New Administrative Capital, Egypt\\
  \texttt{ramy.farouk@eng.asu.edu.eg} \\
}

\begin{document}
\maketitle

\begin{abstract}
Many classical constructions, such as Plotkin's and Turyn's, were generalized by matrix product (MP) codes. Quasi-twisted (QT) codes, on the other hand, form an algebraically rich structure class that contains many codes with best-known parameters. We significantly extend the definition of MP codes to establish a broader class of generalized matrix product (GMP) codes that contains QT codes as well. We propose a generator matrix formula for any linear GMP code and provide a condition for determining the code size. We prove that any QT code has a GMP structure. Then we show how to build a generator polynomial matrix for a QT code from its GMP structure, and vice versa. Despite that the class of QT codes contains many codes with best-known parameters, we present different examples of GMP codes with best-known parameters that are neither MP nor QT. Two different lower bounds on the minimum distance of GMP codes are presented; they generalize their counterparts  in the MP codes literature. The second proposed lower bound replaces the non-singular by columns matrix with a less restrictive condition. Some examples are provided for comparing the two proposed bounds, as well as showing that these bounds are tight.
\end{abstract}

\keywords{Matrix product code \and quasi-twisted code \and best-known parameter \and lower bound}

\section{Introduction}\label{Intro}
In classical coding theory, constructing good long codes from shorter ones is a well-known approach. Several such constructions have been discussed in the literature, including Plotkin's well-known $(u|u + v)$-construction, Turyn's $(a + x|b + x|a + b + x)$-construction \cite{LingS01}, and the ternary $(u + v + w|2u + v|u)$-construction \cite{Kschischang1992}. To generalize such structures, matrix product (MP) codes were originated in \cite{Blackmore2001, zbudak2002}. In \cite{Blackmore2001}, a generator matrix formula for MP codes is proposed, and the code size is determined under certain condition. Furthermore, it has been demonstrated that Reed-Muller codes are iterative MP constructions. Since then, the importance of the class of MP codes began to emerge. 

An MP code of length $m N$ is constructed using $M$ codes $\mathcal{C}_1,\ldots, \mathcal{C}_M$ of length $m$, and an $M\times N$ matrix $\mathcal{A}$. Two different lower bounds on the minimum distance of MP codes have been introduced in the literature. The lower bound introduced in \cite{zbudak2002, Hernando2009} requires that $\mathcal{A}$ be of rank $M$, i.e., $\mathrm{rank}\left(\mathcal{A}\right)=M$. While the lower bound introduced in \cite{Blackmore2001} requires a more restrictive condition, $\mathcal{A}$ must be non-singular by columns (NSC). A proposed extension in the MP construction is studied in \cite{Hernando2010}. In this extension, the matrix $\mathcal{A}$ can have polynomial entries that are coprime to $x^m-1$ while using cyclic codes $\mathcal{C}_1, \ldots, \mathcal{C}_M$. It has been shown that the lower bound introduced in \cite{zbudak2002,Hernando2009} can be naturally generalized to this MP extension.

Cyclic codes have been of particular importance in algebraic coding theory because of their rich algebraic structure. These codes have been generalized in two different ways: generalizing the shift constant results in the class of constacyclic codes, while generalizing the shift index results in the class of quasi-cyclic codes. Although these are two broader classes of cyclic codes, neither of them contains the other. The class of quasi-twisted (QT) codes \cite{Jia2012} was then introduced as a complete class that contains all cyclic, constacyclic, quasi-cyclic codes. Moreover, QT codes have a rich algebraic structure; they can be represented as submodules of free modules over the ring of polynomials, as mentioned in \cite{Taki1}. The class of QT codes is promising not only because of its rich structure, but also because it contains many codes with best-known parameters \cite{QTDatabase,Kasami1974}. It was pointed out in \cite{Cao2020} that there is sometimes a connection between MP and QT codes. It has been shown that some constacyclic codes over finite fields are monomially equivalent to some MP codes with $\mathcal{C}_1,\ldots,\mathcal{C}_M$ constacyclic.

Convinced of the importance of these classes, the main idea of this paper is to offer one class that combines MP and QT codes. Specifically, we establish a new class of codes over finite fields called generalized matrix product (GMP) codes, which has never been studied or even mentioned in the literature. Our definition of GMP codes is a thorough generalization of MP codes. A GMP code $\mathcal{Q}$ over a finite field $\mathbb{F}_q$ of length $m N$ is constructed from $M$ codes $\mathcal{C}_1,\ldots,\mathcal{C}_M$ over $\mathbb{F}_q$ of length $m$, some $M\times N$ matrices over $\mathbb{F}_q$, and a linear transformation on $\mathbb{F}_q^m$. In fact, $\mathcal{Q}$ is a classical MP code if the linear transformation is either the identity map or the zero map. We show that $\mathcal{Q}$ is linear if $\mathcal{C}_1,\ldots,\mathcal{C}_M$ are linear. In this case, we were able to determine a generator matrix formula for $\mathcal{Q}$ that corresponds to that described in \cite{Blackmore2001} for MP codes. We additionally establish a condition under which the size of $\mathcal{Q}$ is the product of the sizes of $\mathcal{C}_1, \ldots,\mathcal{C}_M$. It is a straightforward to see that GMP codes also generalize the MP extension mentioned in \cite{Hernando2010}. 

Section \ref{Sec4} will be devoted to prove that any QT code has actually a GMP structure. This proves our claim that GMP codes combine MP and QT codes into one class. In addition, we show how to derive the GMP structure of a QT code from its generator polynomial matrix, and how to construct a generator polynomial matrix for a QT code written as GMP. Because QT codes contain many codes with best-known parameters, it would not be surprising to find many of these codes in the broader class of GMP codes. Fortunately, GMP codes contain many more best-known parameters codes that are neither MP nor QT. This is shown in Section \ref{BKPexamples}, where we present several GMP codes with different parameters and over different finite fields. All of these codes are neither MP nor QT, but according to the database \cite{codetables}, they achieve the best-known parameters as linear codes. It turns out that it is not necessary to use $\mathcal{C}_1,\ldots,\mathcal{C}_M$ with best-known parameters to obtain $\mathcal{Q}$ with best-known parameters, see Example \ref{ExampleNine}.

For the sake of completeness, we find it necessary to provide lower bounds on the minimum distance of GMP codes. In the first bound, we extend the condition associated with the lower bound of MP codes established in \cite{zbudak2002, Hernando2009} to GMP codes. We present some examples for showing that this bound is attainable and hence tight. We then intend to propose another lower bound that generalizes the one in \cite[Theorem 3.7]{Blackmore2001}. But requiring $\mathcal{A}$ to be NSC is a restrictive condition and would be better replaced by a weaker condition. Theorem \ref{second_bound} provides a lower bound that requires a less restrictive condition than NSC. We show that this lower bound is also tight. Some numerical examples are provided to compare the proposed lower bounds. 

The remainder of this paper is organized as follows. Section \ref{Sec2} covers the MP and QT code structures required for the sequel. Section \ref{Sec3} establishes the GMP structure. It generalizes \cite{Blackmore2001} from MP to GMP codes. In Section \ref{BKPexamples}, we present several examples of GMP codes that are neither MP nor QT but achieve the best-known parameters according to the database \cite{codetables}. In Section \ref{Sec4}, we prove that GMP codes combine MP and QT codes in one class. Section \ref{Sec5} proves the two proposed lower bounds. Finally, we conclude the work in Section \ref{Concl}.

\section{Structures of QT and MP codes}\label{Sec2}
This section reviews the algebraic structure of QT and MP codes, which will be employed in the subsequent sections. For an extensive review of the algebraic structure of QT and MP codes, see \cite{Jia2012, Taki1} and \cite{Blackmore2001, Hernando2009} respectively. A code $\mathcal{C}$ over $\mathbb{F}_q$ of length $m$ is a subset of $\mathbb{F}_q^m$ whose elements are the codewords of $\mathcal{C}$. If $\mathcal{C}$ is a subspace of $\mathbb{F}_q^m$, we call it linear. For a nonzero $\lambda \in\mathbb{F}_q$, a linear code is $\lambda$-constacyclic if it is invariant under the linear transformation that represents the $\lambda$-constacyclic shift 
$$\left(c_0,c_1,\ldots,c_{m-2},c_{m-1}\right)\mapsto \left(\lambda c_{m-1},c_0,c_1,\ldots,c_{m-2}\right).$$
This linear transformation corresponds to left matrix multiplication by $\mathcal{T}_\lambda$ when codewords of $\mathcal{C}$ are represented as $m\times 1$ matrices over $\mathbb{F}_q$, where
\begin{equation}
\label{LambdaTransformation}
\mathcal{T}_\lambda=\begin{bmatrix}
0&0&0&\cdots&0&0&\lambda\\
1&0&0&\cdots&0&0&0\\
0&1&0&\cdots&0&0&0\\
\vdots&\vdots&\vdots&\ddots&\vdots&\vdots&\vdots\\
0&0&0&\cdots&1&0&0\\
0&0&0&\cdots&0&1&0
\end{bmatrix}.
\end{equation}
A $\lambda$-constacyclic code of length $m$ is therefore a $\mathcal{T}_\lambda$-invariant subspace of $\mathbb{F}_q^{m\times 1}\cong \mathbb{F}_q^m$, where $\mathbb{F}_q^{m\times 1}$ is the vector space of all $m\times 1$ matrices over $\mathbb{F}_q$. A cyclic code is constacyclic with $\lambda=1$. 

It is often useful to use a polynomial representation for constacyclic codes. Define the vector space isomorphism $\varphi:\mathbb{F}_q^{m\times 1} \rightarrow \mathcal{R}$, where $\mathcal{R}=\mathbb{F}_q[x]/\langle x^m-\lambda\rangle$ is the quotient of the ring of polynomials over $\mathbb{F}_q$ by the ideal $\langle x^m-\lambda \rangle$. Precisely,
\begin{equation*}
\varphi:\left[ c_0\ c_1\ \cdots\ c_{m-1}\right]^T \mapsto c_0+c_1x+\cdots +c_{m-1}x^{m-1},
\end{equation*}
where $^T$ denotes the matrix transpose. It is only necessary to observe that $\varphi\left( \mathcal{T}_\lambda \mathbf{c}\right)= x\varphi\left( \mathbf{c}\right)$ for any $\mathbf{c}\in\mathbb{F}_q^{m\times 1}$, and hence the polynomial representation of a $\lambda$-constacyclic code $\mathcal{C}$ over $\mathbb{F}_q$ of length $m$ is an ideal of $\mathcal{R}$. Then there exists a generator polynomial $g\left(x\right)\in\mathcal{R}$ for which we can write
$$\mathcal{C}=\langle g\left(x\right)\rangle=\left\{a\left(x\right) g\left(x\right) | a\left(x\right)\in\mathcal{R}\right\}.$$

QT codes generalize constacyclic codes in the same way as quasi-cyclic codes generalize cyclic codes. For a nonzero $\lambda \in \mathbb{F}_q$, a $\lambda$-QT code $\mathcal{Q}$ over $\mathbb{F}_q$ of index $\ell$ and co-index $m$ is a subspace of $\mathbb{F}_q^{m\times \ell}$ that is invariant under left multiplication by $\mathcal{T}_\lambda$, where $\mathbb{F}_q^{m\times \ell}$ is the vector space of all $m\times \ell$ matrices over $\mathbb{F}_q$. In this context, we view codewords of $\mathcal{Q}$ as $m\times \ell$ matrices and say that $\mathcal{Q}$ is a $\mathcal{T}_\lambda$-invariant subspace of $\mathbb{F}_q^{m\times \ell}$. Different algebraic structures for QT codes have been presented in the literature, perhaps the most relevant of which is the polynomial representation. We extend $\varphi$ to a vector space isomorphism between $\mathbb{F}_q^{m\times \ell}$ and $\mathcal{R}^\ell$ as follows
\begin{equation}
\label{MapPhi}
\varphi:\left[ \mathbf{c}_1 \  \mathbf{c}_2 \ \cdots \ \mathbf{c}_\ell \right] \mapsto \left[ \varphi\left(\mathbf{c}_1\right) \  \varphi\left(\mathbf{c}_2\right) \ \cdots \ \varphi\left(\mathbf{c}_\ell\right) \right],
\end{equation}
where $\mathbf{c}_i \in \mathbb{F}_q^{m\times 1}$ and $\varphi\left(\mathbf{c}_i\right) \in \mathcal{R}$ for $1\le i\le \ell$. Again, left multiplication by $\mathcal{T}_\lambda$ in $\mathbb{F}_q^{m\times \ell}$ corresponds to multiplication by $x$ in $\mathcal{R}^\ell$. This is because
\begin{equation}
\label{PhiProperty}
\begin{split}
\varphi\left(\mathcal{T}_\lambda \left[ \mathbf{c}_1 \  \mathbf{c}_2 \ \cdots \ \mathbf{c}_\ell \right]\right) &= \varphi\left( \left[ \mathcal{T}_\lambda\mathbf{c}_1 \  \mathcal{T}_\lambda\mathbf{c}_2 \ \cdots \ \mathcal{T}_\lambda\mathbf{c}_\ell \right]\right)= \left[ \varphi\left(\mathcal{T}_\lambda\mathbf{c}_1\right) \  \varphi\left(\mathcal{T}_\lambda\mathbf{c}_2\right) \ \cdots \ \varphi\left(\mathcal{T}_\lambda\mathbf{c}_\ell\right) \right]\\
&= \left[ x\varphi\left(\mathbf{c}_1\right) \  x\varphi\left(\mathbf{c}_2\right) \ \cdots \  x\varphi\left(\mathbf{c}_\ell\right) \right]=x \left[ \varphi\left(\mathbf{c}_1\right) \  \varphi\left(\mathbf{c}_2\right) \ \cdots \  \varphi\left(\mathbf{c}_\ell\right) \right]\\
&=x \varphi\left(\left[ \mathbf{c}_1 \  \mathbf{c}_2 \ \cdots \   \mathbf{c}_\ell \right]\right).
\end{split}
\end{equation}
This means that the polynomial representation of a $\lambda$-QT code $\mathcal{Q}$ is an $\mathcal{R}$-submodule of $\mathcal{R}^\ell$. Any generating set of $\mathcal{Q}$ as an $\mathcal{R}$-submodule of $\mathcal{R}^\ell$ has a size of at most $\ell$. Elements of a generating set construct rows of a generator polynomial matrix for $\mathcal{Q}$, denoted $\mathfrak{G}$. As $\mathcal{Q}$ may have several generating sets, $\mathfrak{G}$ is not always unique. Furthermore, the size of $\mathfrak{G}$ is not necessary unique; see the paragraph that follows Example \ref{Ex2QTtoGMP}. A QT code with $\lambda=1$ is quasi-cyclic.

Let $\mathcal{C}_1, \mathcal{C}_2, \ldots, \mathcal{C}_M$ be a collection of codes over $\mathbb{F}_q$ of length $m$, and let $\mathcal{A}$ be an $M\times N$ matrix over $\mathbb{F}_q$. The MP code $\mathcal{Q}=\left[ \mathcal{C}_1\  \mathcal{C}_2 \  \cdots \ \mathcal{C}_M\right] \mathcal{A}$ is the subset 
$$\left\{\left[ \mathbf{c}_1\  \mathbf{c}_2 \  \cdots \ \mathbf{c}_M\right] \mathcal{A} \ | \ \mathbf{c}_1\in\mathcal{C}_1, \mathbf{c}_2\in\mathcal{C}_2, \ldots, \mathbf{c}_M\in\mathcal{C}_M\right\}$$
of $\mathbb{F}_q^{m\times N}$. Thus, a typical codeword of $\mathcal{Q}$ is an $m\times N$ matrix whose $(h,l)$-entry is $\sum_{t=1}^{M}\mathrm{Ent}_h\left(\mathbf{c}_t\right)\mathrm{Ent}_{(t,l)}\left(\mathcal{A}\right)$, where $\mathrm{Ent}$ stands for ``Entry''. When $\mathcal{A}$ is the identity matrix $\mathcal{I}_M$, $\mathcal{Q}$ is the direct sum of $\mathcal{C}_1, \mathcal{C}_2, \ldots, \mathcal{C}_M$. In the case when $\mathcal{C}_i$ is linear for every $1\le i\le M$, $\mathcal{Q}$ is linear. Then a generator matrix $\mathbf{G}$ for $\mathcal{Q}$ can be formed by representing the codewords of $\mathcal{Q}$ as vectors in $\mathbb{F}_q^{mN}$ by reading the corresponding matrices in column-major order. Namely,
\begin{equation}
\label{GenMatofMP}
\mathbf{G}=\mathrm{diag}\left[\mathbf{G}_1,\ldots,\mathbf{G}_M\right]  \left( \mathcal{A} \otimes \mathcal{I}_m\right),
\end{equation}
where $\mathrm{diag}\left[\mathbf{G}_1,\ldots,\mathbf{G}_M\right]$ is the block diagonal matrix whose block diagonal entries are the generator matrices of $\mathcal{C}_1,\ldots,\mathcal{C}_M$, and $\otimes$ is the Kronecker product. Formula \eqref{GenMatofMP} is identical to the one described in \cite{Blackmore2001}, however we express it more concisely to be consistent with the subsequent results. This formula will be generalized to GMP codes in Theorem \ref{GenMat}. When $\mathcal{A}$ is right invertible, or, equivalently, $\mathrm{rank}\left(\mathcal{A}\right)=M$, the size $|\mathcal{Q}|$ of $\mathcal{Q}$ is the product $\prod_{i=1}^{M}|\mathcal{C}_i|$.

An $M\times N$ matrix $\mathcal{A}$ over $\mathbb{F}_q$ with $M \le N$ is called non-singular by columns (NSC) if for every $1\le t\le M$ and every choice $1\le j_1 < j_2 <\cdots <j_t \le N$ the $t\times t$ submatrix $\mathcal{A}^{(t)}\left(j_1, j_2, \ldots, j_t \right)$ of $\mathcal{A}$ consisting of the first $t$ rows and the $j_1, j_2, \ldots, j_t$ columns is non-singular. In \cite{Blackmore2001}, a lower bound on the minimum distance $d\left(\mathcal{Q}\right)$ of an MP code $\mathcal{Q}=\left[ \mathcal{C}_1\  \mathcal{C}_2 \  \cdots \ \mathcal{C}_M\right] \mathcal{A}$ with NSC $\mathcal{A}$ is given by
\begin{equation}
\label{bound1}
d\left(\mathcal{Q}\right)\ge \min \left\{ N d\left(\mathcal{C}_1\right), \left(N-1\right)d\left(\mathcal{C}_2\right), \ldots, \left(N-M+1\right)d\left(\mathcal{C}_M\right)\right\}.
\end{equation}
Being NSC is a restrictive requirement on $\mathcal{A}$, which has been waived in the lower bound proposed in \cite{zbudak2002,Hernando2009}. This lower bound only requires $\mathcal{A}$ to be of full rank, which is more convenient since most MP codes have matrices that are not necessarily NSC. Specifically, if $\mathrm{rank}\left(\mathcal{A}\right)=M$, then
\begin{equation}
\label{bound2}
d\left(\mathcal{Q}\right)\ge \min\left\{ D_1 d\left(\mathcal{C}_1\right), D_2 d\left(\mathcal{C}_2\right), \ldots, D_M d\left(\mathcal{C}_M\right)\right\},
\end{equation}
where $D_t$ is the minimum distance of the linear code over $\mathbb{F}_q$ of length $N$ spanned by the $t\times N$ submatrix $\mathcal{A}^{(t)}$. This lower bound is shown in \cite{Hernando2010} to be applicable in the MP extension in which $\mathcal{A}$ has polynomial entries coprime to $x^m-1$ rather than elements in $\mathbb{F}_q$ and $\mathcal{C}_1,\ldots,\mathcal{C}_M$ are cyclic.

\section{Generalized MP codes}
\label{Sec3}
As a generalization of MP codes, we mainly investigate the class of GMP codes in this section. This class has not been introduced previously in the literature. A GMP code $\mathcal{Q}$ of length $mN$ is constructed from $M$ codes $\mathcal{C}_i$ of length $m$, some matrices of size $M\times N$, and a linear transformation on $\mathbb{F}_q^m$. Codewords of $\mathcal{C}_i$ and $\mathcal{Q}$ are, respectively, $m\times 1$ and $m\times N$ matrices over $\mathbb{F}_q$.
\begin{definition}
\label{GMP_Def}
Let $\mathcal{T}$ be the matrix of a linear transformation on $\mathbb{F}_q^m$ with respect to the standard basis, let $\mathcal{C}_1,\ldots,\mathcal{C}_M$ be a collection of codes over $\mathbb{F}_q$ of length $m$, and let $\mathcal{A}_0,\ldots, \mathcal{A}_r$ be some $M\times N$ matrices over $\mathbb{F}_q$. The GMP code $\mathcal{Q}=\sum_{k=0}^{r} \mathcal{T}^k \left[ \mathcal{C}_1 \  \cdots \ \mathcal{C}_M\right] \mathcal{A}_k$ is the subset of $\mathbb{F}_q^{m \times N}$ that contains all matrices of the form $\sum_{k=0}^{r} \mathcal{T}^k \left[ \mathbf{c}_1 \ \cdots \ \mathbf{c}_M\right] \mathcal{A}_k$, where $\mathbf{c}_i\in\mathcal{C}_i$ for $1\le i\le M$.
\end{definition}
For any $m\times m$ matrix $\mathcal{T}$, we follow the convention that $\mathcal{T}^0=\mathcal{I}_m$. A GMP code is therefore written as 
\begin{equation*}
\begin{split}
\mathcal{Q}= \sum_{k=0}^{r} \mathcal{T}^k \left[ \mathcal{C}_1 \ \cdots \ \mathcal{C}_M\right] \mathcal{A}_k
= \left[ \mathcal{C}_1 \ \cdots \  \mathcal{C}_M\right] \mathcal{A}_0+\sum_{k=1}^{r} \mathcal{T}^k \left[ \mathcal{C}_1 \ \cdots \ \mathcal{C}_M\right] \mathcal{A}_k.
\end{split}
\end{equation*}
This broadens the definition of MP codes in \cite[Definition 2.1]{Blackmore2001} because $\mathcal{Q}$ is MP if $\mathcal{T}=\mathcal{I}_m$, $\mathcal{T}=\mathbf{0}_{m\times m}$, or $\mathcal{A}_1=\cdots=\mathcal{A}_r=\mathbf{0}_{M\times N}$, where $\mathbf{0}_{M\times N}$ is the zero matrix of size $M\times N$. Specifically, $\mathcal{Q}=  \left[ \mathcal{C}_1 \ \cdots \  \mathcal{C}_M\right]\mathcal{A}$, where $\mathcal{A}= \sum_{k=0}^{r}\mathcal{A}_k$ when $\mathcal{T}=\mathcal{I}_m$, and $\mathcal{A}=\mathcal{A}_0$ when $\mathcal{T}=\mathbf{0}_{m\times m}$ or $\mathcal{A}_1=\cdots=\mathcal{A}_r=\mathbf{0}_{M\times N}$. Definition \ref{GMP_Def} places GMP codes in a very broad class, which generalizes not only MP codes, but also the MP extension described in \cite{Hernando2010} and QT codes as we shall show in Section \ref{Sec4}.

We now aim to construct a generator matrix for any linear GMP code from the generator matrices of $\mathcal{C}_1,\ldots,\mathcal{C}_M$. To this end, we express codewords of $\mathcal{Q}$ in vector form rather than matrix form using the vectorization map. We refer to the vectorization of matrices in column-major order by $\mathrm{vec}$; it defines a vector space isomorphism between $\mathbb{F}_q^{m\times N}$ and $\mathbb{F}_q^{m N}$. For instance, the vectorization of the $m\times N$ matrix $\left[a_{i,j}\right]\in \mathbb{F}_q^{m\times N}$ is
\begin{equation*}
\mathrm{vec}\left(\left[a_{i,j}\right]\right)= \left( a_{1,1},\ldots,a_{m,1},a_{1,2},\ldots,a_{m,2},\ldots,a_{1,N},\ldots,a_{m,N}\right) \in \mathbb{F}_q^{mN}.
\end{equation*}
It is now clear that $\mathrm{vec}\left(\mathcal{Q}\right)$ is a code of length $mN$, which if linear would be isomorphic to $\mathcal{Q}$. In general, distinguishing between $\mathcal{Q}$ and $\mathrm{vec}\left(\mathcal{Q}\right)$ is unnecessary, because what we mean can be determined from context. We show that $\mathcal{Q}$ (and, hence, $\mathrm{vec}\left(\mathcal{Q}\right)$) is linear if $\mathcal{C}_1,\ldots,\mathcal{C}_M$ are linear. This is simply by observing that, for any $a,b\in\mathbb{F}_q$ and $\mathbf{a}_i,\mathbf{b}_i\in\mathcal{C}_i$,
\begin{equation*}\begin{split}
&a\sum_{k=0}^{r} \mathcal{T}^k \left[ \mathbf{a}_1 \ \cdots \ \mathbf{a}_M\right] \mathcal{A}_k+b\sum_{k=0}^{r} \mathcal{T}^k \left[ \mathbf{b}_1 \ \cdots \ \mathbf{b}_M\right] \mathcal{A}_k=\\&\qquad\qquad\qquad \sum_{k=0}^{r} \mathcal{T}^k \left[ \left(a\mathbf{a}_1+b\mathbf{b}_1\right) \ \cdots \ \left(a\mathbf{a}_M+b\mathbf{b}_M\right)\right] \mathcal{A}_k \in \mathcal{Q}.
\end{split}\end{equation*}
If $\mathcal{Q}$ is linear, it is legitimate to ask about a generator matrix for $\mathrm{vec}\left(\mathcal{Q}\right)$. The following theorem offers such a generator matrix which is not necessary of full rank. 

\begin{theorem}
\label{GenMat}
Let $\mathcal{Q}=\sum_{k=0}^{r} \mathcal{T}^k \left[ \mathcal{C}_1 \ \cdots \  \mathcal{C}_M\right] \mathcal{A}_k$ for some linear codes $\mathcal{C}_1,\ldots,\mathcal{C}_M$ with generator matrices $\mathbf{G}_1, \ldots,\mathbf{G}_M$. Then a generator matrix for $\mathrm{vec}\left(\mathcal{Q}\right)$ is given by
\begin{equation*}
\mathbf{G}=\mathrm{diag}\left[\mathbf{G}_1,\ldots,\mathbf{G}_M\right] \sum_{k=0}^{r} \mathcal{A}_k \otimes \left(\mathcal{T}^T\right)^k.
\end{equation*}
\end{theorem}
\begin{proof}
Let $\mathbf{q}=\sum_{k=0}^{r} \mathcal{T}^k \left[ \mathbf{c}_1 \ \cdots \ \mathbf{c}_M\right] \mathcal{A}_k$ be an arbitrary codeword of $\mathcal{Q}$. Let $k_i\times m$ be the size of $\mathbf{G}_i$ with $m$ being the length of $\mathcal{C}_i$. Since each $\mathbf{c}_i\in\mathbb{F}_q^{m\times 1}$ is a codeword of $\mathcal{C}_i$, there exists $\mathbf{a}_i \in \mathbb{F}_q^{1\times k_i}$ such that $\mathbf{c}_i^T=\mathbf{a}_i \mathbf{G}_i$. Then
\begin{equation*}\begin{split}
\mathrm{vec}\left(\mathbf{q}\right)&=\mathrm{vec}\left(\sum_{k=0}^{r} \mathcal{T}^k \left[ \mathbf{c}_1 \ \cdots \ \mathbf{c}_M\right] \mathcal{A}_k \right)=\sum_{k=0}^{r} \mathrm{vec}\left( \mathcal{T}^k \left[ \mathbf{c}_1 \ \cdots \ \mathbf{c}_M\right] \mathcal{A}_k \right)\\
&= \mathrm{vec}\left(\left[ \mathbf{c}_1 \ \cdots \ \mathbf{c}_M\right]\right) \sum_{k=0}^{r} \mathcal{A}_k \otimes \left(\mathcal{T}^T\right)^k
= \left[ \mathbf{a}_1\mathbf{G}_1 \ \cdots \ \mathbf{a}_M \mathbf{G}_M\right] \sum_{k=0}^{r} \mathcal{A}_k \otimes \left(\mathcal{T}^T\right)^k\\
&= \left[ \mathbf{a}_1 \ \cdots \ \mathbf{a}_M \right] \mathrm{diag}\left[\mathbf{G}_1,\ldots,\mathbf{G}_M\right] \sum_{k=0}^{r} \mathcal{A}_k  \otimes \left(\mathcal{T}^T\right)^k
= \left[ \mathbf{a}_1 \ \cdots  \ \mathbf{a}_M \right] \mathbf{G}.
\end{split}\end{equation*}
This shows that codewords of $\mathrm{vec}\left(\mathcal{Q}\right)$ are in the row span of $\mathbf{G}$. Conversely, $\mathrm{vec}\left(\mathcal{Q}\right)$ contains the row span of $\mathbf{G}$ because
\begin{equation*}
\left[ \mathbf{a}_1 \ \cdots  \  \mathbf{a}_M \right] \mathbf{G}= \mathrm{vec}\left(\sum_{k=0}^{r} \mathcal{T}^k \left[ \mathbf{c}_1 \ \cdots  \  \mathbf{c}_M\right] \mathcal{A}_k\right) \in \mathrm{vec}\left(\mathcal{Q}\right)
\end{equation*}
for each choice of $\mathbf{a}_i \in \mathbb{F}_q^{1\times k_i}$ by letting $\mathbf{c}_i^T=\mathbf{a}_i \mathbf{G}_i$.
\end{proof}

Theorem \ref{GenMat} generalizes the generator matrix formula of MP codes given in \eqref{GenMatofMP}, e.g., by setting $\mathcal{A}_1=\cdots=\mathcal{A}_r=\mathbf{0}_{M\times N}$. The following theorem establishes a condition under which the size of a GMP code can be easily determined. This yields the analogous conclusion to \cite[Proposition 2.9]{Blackmore2001} for the more general class of GMP codes.
\begin{theorem}
\label{GenMat_dim}
Let $\mathcal{Q}=\sum_{k=0}^{r} \mathcal{T}^k \left[ \mathcal{C}_1 \ \cdots  \   \mathcal{C}_M\right] \mathcal{A}_k$ be a GMP code and assume that $\sum_{k=0}^{r} \mathcal{A}_k \otimes \left(\mathcal{T}^T\right)^k$ has rank $mM$. Then $$|\mathcal{Q}|=\prod_{i=1}^{M}|\mathcal{C}_i|.$$ 
If $\mathcal{C}_1,\ldots,\mathcal{C}_M$ are linear, then $\mathrm{dim}\left(\mathcal{Q}\right)=\sum_{i=1}^M \mathrm{dim}\left(\mathcal{C}_i\right)$, where $\mathrm{dim}$ stands for ``dimension''.
\end{theorem}
\begin{proof}
The mapping
$$\mathrm{vec}\left(\left[ \mathbf{c}_1 \ \cdots  \  \mathbf{c}_M\right]\right) \mapsto \mathrm{vec}\left(\left[ \mathbf{c}_1 \ \cdots  \  \mathbf{c}_M\right]\right) \sum_{k=0}^{r} \mathcal{A}_k \otimes \left(\mathcal{T}^T\right)^k$$
takes each choice of codewords $\mathbf{c}_1\in\mathcal{C}_1, \ldots , \mathbf{c}_M\in\mathcal{C}_M$ to the codeword $\mathrm{vec}\left(\sum_{k=0}^{r} \mathcal{T}^k \left[ \mathbf{c}_1 \ \cdots  \  \mathbf{c}_M\right] \mathcal{A}_k \right)$ of $\mathrm{vec}\left(\mathcal{Q}\right)$. Now, if this map is injective, then $|\mathcal{Q}|=\prod_{i=1}^{M}|\mathcal{C}_i|$ . But this map is injective if $\sum_{k=0}^{r} \mathcal{A}_k \otimes \left(\mathcal{T}^T\right)^k$ is right invertible. The latter, however, is true if and only if $\sum_{k=0}^{r} \mathcal{A}_k \otimes \left(\mathcal{T}^T\right)^k$ has rank $mM$.
\end{proof}

By considering Theorem \ref{GenMat_dim} in the particular setting of  $\mathcal{A}_1=\cdots=\mathcal{A}_r=\mathbf{0}_{M\times N}$, the condition that $\sum_{k=0}^{r} \mathcal{A}_k \otimes \left(\mathcal{T}^T\right)^k$ has rank $mM$ simplifies to $\mathcal{A}_0$ having rank $M$. In fact, this is the result on the size of MP codes that was noticed in \cite{Blackmore2001}. Additionally, we remark that the condition of Theorem \ref{GenMat_dim} is sufficient but not necessary. This can be observed via Example \ref{ExampleEight}, where $|\mathcal{Q}|=\prod_{i=1}^{M}|\mathcal{C}_i|$ even though $\sum_{k=0}^{r} \mathcal{A}_k \otimes \left(\mathcal{T}^T\right)^k$ has rank less than $mM$. In contrast, $|\mathcal{Q}| \ne \prod_{i=1}^{M}|\mathcal{C}_i|$ in Example \ref{ExampleFive} while $\sum_{k=0}^{r} \mathcal{A}_k \otimes \left(\mathcal{T}^T\right)^k$ has rank less than $mM$.

\section{GMP codes with best-known parameters}
\label{BKPexamples}
In Section \ref{Sec4}, we will prove that QT codes have a GMP structure. However, QT codes are known to include many good codes. As a result, we expect to find many GMP codes with best-known parameters. In this section, many of these codes---with different parameters and over different fields---will be presented. To make these codes noteworthy, all of them are neither MP nor QT. 

\begin{example}
\label{ExampleThree}
Let $m=4$, $M=2$, $N=5$, and $r=1$. Consider the binary GMP code $\mathcal{Q}= \left[ \mathcal{C}_1 \ \mathcal{C}_2\right] \mathcal{A}_0+ \mathcal{T} \left[ \mathcal{C}_1 \  \mathcal{C}_2\right] \mathcal{A}_1$ of length $20$, where
\begin{equation*}
\mathcal{T}=\begin{bmatrix}
0 & 0 & 0 &1\\
0 & 1 & 1 &1\\
0 & 1 & 0 &1\\
0 & 0 & 1 &1
\end{bmatrix},
\quad
\mathcal{A}_0=\begin{bmatrix}
0 & 1 & 1 &1 &0\\
0 & 0 & 1 &0 &1
\end{bmatrix},
\quad
\mathcal{A}_1=\begin{bmatrix}
1 & 1 & 1 &1 &0\\
1 & 0 & 0 &1 &1
\end{bmatrix},
\end{equation*}
and $\mathcal{C}_1=\mathcal{C}_2$ is the binary linear code with generator matrix
\begin{equation*}
\mathbf{G}_1=\mathbf{G}_2=\begin{bmatrix}
1 & 0 & 0 &1\\
0 & 1 & 0 &1\\
0 & 0 & 1 &1
\end{bmatrix}.
\end{equation*}
The linearity of $\mathcal{C}_1$ and $\mathcal{C}_2$ implies the linearity of $\mathcal{Q}$. From Theorem \ref{GenMat}, $\mathrm{vec}\left(\mathcal{Q}\right)$ has the generator matrix
\setcounter{MaxMatrixCols}{20}\begin{equation*}\begin{split}
\mathbf{G}=\mathrm{diag}\left[\mathbf{G}_1,\mathbf{G}_1\right] \left(\mathcal{A}_0 \otimes \mathcal{I}_4 + \mathcal{A}_1 \otimes \mathcal{T}^T \right)
=\begin{bmatrix}
1&1&1&1&0&1&1&0&0&1&1&0&0&1&1&0&0&0&0&0\\
1&0&0&1&1&1&0&0&1&1&0&0&1&1&0&0&0&0&0&0\\
1&0&1&0&1&0&0&1&1&0&0&1&1&0&0&1&0&0&0&0\\
1&1&1&1&0&0&0&0&1&0&0&1&1&1&1&1&0&1&1&0\\
1&0&0&1&0&0&0&0&0&1&0&1&1&0&0&1&1&1&0&0\\
1&0&1&0&0&0&0&0&0&0&1&1&1&0&1&0&1&0&0&1
\end{bmatrix}.
\end{split}\end{equation*}
One can verify that $\mathrm{rank}\left(\mathcal{A}_0 \otimes \mathcal{I}_4 + \mathcal{A}_1 \otimes \mathcal{T}^T \right) =8=mM$, and hence Theorem \ref{GenMat_dim} implies that $\mathrm{dim}\left(\mathcal{Q}\right)=\mathrm{dim}\left(\mathcal{C}_1\right)+\mathrm{dim}\left(\mathcal{C}_2\right)=6$. The minimum distance $d\left(\mathcal{Q}\right)$ of $\mathcal{Q}$ is determined and found to be $d\left(\mathcal{Q}\right)=8$. According to \cite{codetables}, $\mathcal{Q}$ has the best-known parameters $[20,6,8]_2$.
\end{example}

\begin{example}
\label{ExampleFour}
Let $m=5$, $M=3$, $N=6$, $r=2$, and $\alpha\in\mathbb{F}_4$ such that $\alpha^2+\alpha+1=0$. Consider the GMP code $$\mathcal{Q}= \left[ \mathcal{C}_1 \ \mathcal{C}_2\ \mathcal{C}_3\right] \mathcal{A}_0+ \mathcal{T} \left[ \mathcal{C}_1 \  \mathcal{C}_2 \ \mathcal{C}_3\right] \mathcal{A}_1+ \mathcal{T}^2 \left[ \mathcal{C}_1 \  \mathcal{C}_2 \ \mathcal{C}_3\right] \mathcal{A}_2$$ over $\mathbb{F}_4$ of length $30$, where
\begin{align*}
\mathcal{T}&=\begin{bmatrix}
0&1&0&\alpha&\alpha^2\\
\alpha&1&1&0&0\\
\alpha&0&1&\alpha^2&0\\
1&1&\alpha^2&\alpha^2&\alpha\\
\alpha&\alpha^2&0&\alpha&\alpha
\end{bmatrix},
\quad
\mathcal{A}_0=\begin{bmatrix}
   0&  0&\alpha & 1&  0&  0\\
   1& \alpha& \alpha&\alpha&\alpha&  0\\
   1&  1&\alpha& \alpha^2&  1&  1
\end{bmatrix},\\
\mathcal{A}_1&=\begin{bmatrix}
 \alpha^2&\alpha&  1&\alpha & 1&\alpha^2\\
   1&\alpha & 0&  0&\alpha^2 & 1\\
 \alpha&\alpha & 1&  0 & 0&  1
\end{bmatrix},
\quad
\mathcal{A}_2=\begin{bmatrix}
   1&\alpha^2 &\alpha^2&\alpha&\alpha & 0\\
   1&\alpha & 0&  0&\alpha^2&\alpha^2\\
 \alpha&  1&  1&\alpha & 0&\alpha
\end{bmatrix},
\end{align*}
and $\mathcal{C}_1=\mathcal{C}_2=\mathcal{C}_3$ is the repetition code over $\mathbb{F}_4$ of length $5$. That is,
\begin{equation*}
\mathbf{G}_1=\mathbf{G}_2=\mathbf{G}_3=\begin{bmatrix}
1&1&1&1&1
\end{bmatrix}.
\end{equation*}
The linearity of $\mathcal{C}_1$, $\mathcal{C}_2$, and $\mathcal{C}_3$ implies the linearity of $\mathcal{Q}$. From Theorem \ref{GenMat}, $\mathrm{vec}\left(\mathcal{Q}\right)$ has the generator matrix
\begin{equation*}
\mathbf{G}=\mathrm{diag}\left[\mathbf{G}_1,\mathbf{G}_1,\mathbf{G}_1\right] \left(\mathcal{A}_0 \otimes \mathcal{I}_5 + \mathcal{A}_1 \otimes \mathcal{T}^T+ \mathcal{A}_2 \otimes \left(\mathcal{T}^T\right)^2 \right).
\end{equation*}
One can verify that $\mathrm{rank}\left(\mathcal{A}_0 \otimes \mathcal{I}_5 + \mathcal{A}_1 \otimes \mathcal{T}^T+ \mathcal{A}_2 \otimes \left(\mathcal{T}^T\right)^2 \right)=15=mM$, and hence Theorem \ref{GenMat_dim} implies that $\mathrm{dim}\left(\mathcal{Q}\right)=3 \mathrm{dim}\left(\mathcal{C}_1\right)=3$. The minimum distance of $\mathcal{Q}$ is $d\left(\mathcal{Q}\right)=22$. According to \cite{codetables}, $\mathcal{Q}$ has the best-known parameters $[30,3,22]_4$.
\end{example}

\begin{example}
\label{ExampleFive}
Let $m=M=4$, $N=3$, and $r=1$. Since $M>N$, Theorem \ref{GenMat_dim} is inapplicable. In other words, the rank of the $mM \times mN$ matrix $\sum_{k=0}^{r} \mathcal{A}_k \otimes \left(\mathcal{T}^T\right)^k$ is at most $mN<mM$, and we cannot ensure that $\mathrm{dim}\left(\mathcal{Q}\right)=\sum_{i=1}^M \mathrm{dim}\left(\mathcal{C}_i\right)$. Let $\mathcal{Q}$ be the GMP code $$\mathcal{Q}= \left[ \mathcal{C}_1 \ \mathcal{C}_2\ \mathcal{C}_3\ \mathcal{C}_4\right] \mathcal{A}_0+ \mathcal{T} \left[ \mathcal{C}_1 \  \mathcal{C}_2 \ \mathcal{C}_3 \ \mathcal{C}_4\right] \mathcal{A}_1$$ over $\mathbb{F}_3$ of length $12$, where
\begin{equation*}
\mathcal{T}=\begin{bmatrix}
2&2&0&0\\
0&2&2&1\\
0&1&0&0\\
2&2&1&2
\end{bmatrix},
\quad
\mathcal{A}_0=\begin{bmatrix}
2&2&2\\
0&2&2\\
0&1&2\\
2&0&2
\end{bmatrix},
\quad
\mathcal{A}_1=\begin{bmatrix}
1&2&2\\
0&0&0\\
0&2&2\\
0&1&1
\end{bmatrix},
\end{equation*}
$\mathcal{C}_1=\mathcal{C}_2$ is the linear code over $\mathbb{F}_3$ with generator matrix 
\begin{equation*}
\mathbf{G}_1=\mathbf{G}_2=\begin{bmatrix}
1&0&2&2\\
0&1&1&2
\end{bmatrix},
\end{equation*}
and $\mathcal{C}_3=\mathcal{C}_4$ is the linear code over $\mathbb{F}_3$ with generator matrix 
\begin{equation*}
\mathbf{G}_3=\mathbf{G}_4=\begin{bmatrix}
1&0&0&1\\
0&1&0&1\\
0&0&1&1
\end{bmatrix}.
\end{equation*}
The linearity of $\mathcal{C}_1$, $\mathcal{C}_2$, $\mathcal{C}_3$, and $\mathcal{C}_4$ implies the linearity of $\mathcal{Q}$. From Theorem \ref{GenMat}, $\mathrm{vec}\left(\mathcal{Q}\right)$ has the generator matrix
\begin{equation*}
\mathbf{G}=\mathrm{diag}\left[\mathbf{G}_1,\mathbf{G}_1,\mathbf{G}_3,\mathbf{G}_3\right] \left(\mathcal{A}_0 \otimes \mathcal{I}_4 + \mathcal{A}_1 \otimes \mathcal{T}^T\right).
\end{equation*}
The dimension of $\mathcal{Q}$ is $\mathrm{dim}\left(\mathcal{Q}\right)=9$, while its minimum distance is $d\left(\mathcal{Q}\right)=3$. According to \cite{codetables}, $\mathcal{Q}$ has the best-known parameters $[12,9,3]_3$.
\end{example}

\begin{example}
\label{ExampleSeven}
Let $m=M=3$, $N=4$, and $r=3$. Consider the GMP code $$\mathcal{Q}= \left[ \mathcal{C}_1 \ \mathcal{C}_2\ \mathcal{C}_3\right] \mathcal{A}_0+ \mathcal{T} \left[ \mathcal{C}_1 \  \mathcal{C}_2 \ \mathcal{C}_3\right] \mathcal{A}_1 + \mathcal{T}^2 \left[ \mathcal{C}_1 \  \mathcal{C}_2 \ \mathcal{C}_3\right] \mathcal{A}_2 + \mathcal{T}^3 \left[ \mathcal{C}_1 \  \mathcal{C}_2 \ \mathcal{C}_3\right] \mathcal{A}_3$$ over $\mathbb{F}_5$ of length $12$, where
\begin{equation*}\begin{split}
\mathcal{T}=\begin{bmatrix}
2&2&3\\
0&3&1\\
2&0&0
\end{bmatrix},
\quad
\mathcal{A}_0=\begin{bmatrix}
1&0&3&1\\
4&0&2&2\\
4&2&1&1
\end{bmatrix},
\quad
\mathcal{A}_1=\begin{bmatrix}
2&0&1&1\\
3&4&0&2\\
3&0&2&3
\end{bmatrix},\\
\mathcal{A}_2=\begin{bmatrix}
0&0&2&2\\
1&0&4&4\\
4&3&4&3
\end{bmatrix},
\quad
\mathcal{A}_3=\begin{bmatrix}
1&1&1&3\\
3&4&1&2\\
0&3&4&0
\end{bmatrix},
\end{split}\end{equation*}
$\mathcal{C}_1$ is the repetition code, i.e., $\mathbf{G}_1=\left[1\ 1\ 1 \right]$, and $\mathcal{C}_2=\mathcal{C}_3$ is the linear code with generator matrix 
\begin{equation*}
\mathbf{G}_2=\mathbf{G}_3=\begin{bmatrix}
1&0&4\\
0&1&4
\end{bmatrix}.
\end{equation*}
The linearity of $\mathcal{C}_1$, $\mathcal{C}_2$, and $\mathcal{C}_3$ implies the linearity of $\mathcal{Q}$. From Theorem \ref{GenMat}, $\mathrm{vec}\left(\mathcal{Q}\right)$ has the generator matrix
\begin{equation*}
\mathbf{G}=\mathrm{diag}\left[\mathbf{G}_1,\mathbf{G}_2,\mathbf{G}_2\right] \left(\mathcal{A}_0 \otimes \mathcal{I}_3 + \mathcal{A}_1 \otimes \mathcal{T}^T + \mathcal{A}_2 \otimes \left(\mathcal{T}^T\right)^2+ \mathcal{A}_3 \otimes \left(\mathcal{T}^T\right)^3\right).
\end{equation*}
Since $\mathrm{rank}\left(\mathcal{A}_0 \otimes \mathcal{I}_3 + \mathcal{A}_1 \otimes \mathcal{T}^T + \mathcal{A}_2 \otimes \left(\mathcal{T}^T\right)^2+ \mathcal{A}_3 \otimes \left(\mathcal{T}^T\right)^3\right)=9=mM$, Theorem \ref{GenMat_dim} implies that $\mathrm{dim}\left(\mathcal{Q}\right)=\mathrm{dim}\left(\mathcal{C}_1\right)+2\mathrm{dim}\left(\mathcal{C}_2\right)=5$. We found $d\left(\mathcal{Q}\right)=6$ and, hence, $\mathcal{Q}$ has the best-known parameters $[12,5,6]_5$ according to \cite{codetables}.
\end{example}

\begin{example}
\label{ExampleEight}
This example shows that the rank condition in Theorem \ref{GenMat_dim} is sufficient but not necessary. Specifically, we will conclude that $\mathrm{dim}\left(\mathcal{Q}\right)=\sum_{i=1}^M \mathrm{dim}\left(\mathcal{C}_i\right)$ even though $\mathrm{rank}\left(\sum_{k=0}^{r} \mathcal{A}_k \otimes \left(\mathcal{T}^T\right)^k \right) \ne mM$. Let $m=6$, $M=N=2$, $r=1$, and $\alpha\in\mathbb{F}_9$ such that $\alpha^2+2\alpha+2=0$. Consider the GMP code $$\mathcal{Q}= \left[ \mathcal{C}_1 \ \mathcal{C}_2 \right] \mathcal{A}_0+ \mathcal{T} \left[ \mathcal{C}_1 \  \mathcal{C}_2 \right] \mathcal{A}_1$$ over $\mathbb{F}_9$ of length $12$, where
\begin{equation*}\begin{split}
\mathcal{T}=\begin{bmatrix}
 \alpha&  2&\alpha^2&\alpha^5&  0&  1\\
 \alpha^7&  1&\alpha^3&\alpha^2&  2&  2\\
   0&\alpha^6&\alpha^5&  0&\alpha&\alpha^2\\
   1&\alpha^6&\alpha^3&  2&\alpha^2&  2\\
 \alpha^7&\alpha^3&\alpha&\alpha^5&\alpha^2&\alpha^5\\
 \alpha^3&  0&  0&\alpha^3&\alpha^2&  0
\end{bmatrix},
\quad
\mathcal{A}_0=\begin{bmatrix}
2& \alpha^2\\
 \alpha^6 &\alpha^7
\end{bmatrix},
\quad
\mathcal{A}_1=\begin{bmatrix}
\alpha^2 &\alpha^7\\
\alpha& \alpha^2
\end{bmatrix},
\end{split}\end{equation*}
$\mathcal{C}_1$ is the $[6,2,5]_9$-MDS code generated by 
\begin{equation*}
\mathbf{G}_1=\begin{bmatrix}
1& 0&\alpha^3&\alpha^3&  1&\alpha^6\\
0&  1&\alpha^7&\alpha^6&\alpha^5&\alpha^6
\end{bmatrix},
\end{equation*}
and $\mathcal{C}_2$ is the $[6,4,3]_9$-MDS code generated by 
\begin{equation*}
\mathbf{G}_2=\begin{bmatrix}
1&  0&  0&  0&\alpha^7&\alpha\\
0  &1&  0&  0&\alpha^3&  1\\
0&  0&  1&  0&\alpha&  2\\
0&  0&  0&  1&\alpha^5&\alpha^5
\end{bmatrix}.
\end{equation*}
From Theorem \ref{GenMat}, $\mathrm{vec}\left(\mathcal{Q}\right)$ has the generator matrix
\begin{equation*}
\mathbf{G}=\mathrm{diag}\left[\mathbf{G}_1,\mathbf{G}_2\right] \left(\mathcal{A}_0 \otimes \mathcal{I}_6 + \mathcal{A}_1 \otimes \mathcal{T}^T \right).
\end{equation*}
Although $\mathrm{rank}\left(\mathcal{A}_0 \otimes \mathcal{I}_6 + \mathcal{A}_1 \otimes \mathcal{T}^T \right)=11 < mM$, we found $\mathrm{dim}\left(\mathcal{Q}\right)=6=\mathrm{dim}\left(\mathcal{C}_1\right)+\mathrm{dim}\left(\mathcal{C}_2\right)$. Moreover, we found $d\left(\mathcal{Q}\right)=6$ and, hence, $\mathcal{Q}$ has the best-known parameters $[12,6,6]_9$ according to \cite{codetables}.
\end{example}

Referring to Example \ref{ExampleEight}, we note that $\mathcal{Q}$ achieves best-known parameters and is built from codes $\mathcal{C}_i$ that likewise have best-known parameters. However, as the following example will show, a GMP code $\mathcal{Q}$ with best-known parameters does not have to be built from codes $\mathcal{C}_i$ with best-known parameters.

\begin{example}
\label{ExampleNine}
Let $m=M=3$, $N=6$, $r=1$, and $\alpha\in\mathbb{F}_9$ such that $\alpha^2+2\alpha+2=0$. Consider the GMP code $$\mathcal{Q}= \left[ \mathcal{C}_1 \ \mathcal{C}_2 \ \mathcal{C}_3 \right] \mathcal{A}_0+ \mathcal{T} \left[ \mathcal{C}_1 \  \mathcal{C}_2 \ \mathcal{C}_3 \right] \mathcal{A}_1$$ over $\mathbb{F}_9$ of length $18$, where
\begin{equation*}\begin{split}
\mathcal{T}=\begin{bmatrix}
\alpha^6&  \alpha^2&\alpha^3\\
2&  2&\alpha^6\\
\alpha^3&\alpha^7&\alpha
\end{bmatrix},
\quad
\mathcal{A}_0=\begin{bmatrix}
0&\alpha^3&\alpha&  2&\alpha^6&  2\\
2&\alpha^7&  1&  0&\alpha^7&\alpha^5\\
\alpha&\alpha^3&  2&\alpha^5&\alpha&\alpha^5
\end{bmatrix},
\quad
\mathcal{A}_1=\begin{bmatrix}
1&\alpha^6&\alpha^7&\alpha^2&\alpha^2&  1\\
\alpha^7&\alpha^2&\alpha^5&  2&\alpha^6&  0\\
\alpha&\alpha^5&\alpha^7&  1&\alpha^2&\alpha^3
\end{bmatrix},
\end{split}\end{equation*}
$\mathcal{C}_1=\mathcal{C}_2$ is the $[3,1,2]_9$-code generated by 
\begin{equation*}
\mathbf{G}_1=\mathbf{G}_2=\begin{bmatrix}
0&1&\alpha^2
\end{bmatrix},
\end{equation*}
and $\mathcal{C}_3$ is the $[3,2,2]_9$ code generated by 
\begin{equation*}
\mathbf{G}_3=\begin{bmatrix}
1&0&2\\
0&1&\alpha
\end{bmatrix}.
\end{equation*}
From Theorem \ref{GenMat}, $\mathrm{vec}\left(\mathcal{Q}\right)$ has the generator matrix
\begin{equation*}
\mathbf{G}=\mathrm{diag}\left[\mathbf{G}_1,\mathbf{G}_1,\mathbf{G}_3\right] \left(\mathcal{A}_0 \otimes \mathcal{I}_3 + \mathcal{A}_1 \otimes \mathcal{T}^T \right).
\end{equation*}
Since $\mathrm{rank}\left(\mathcal{A}_0 \otimes \mathcal{I}_3 + \mathcal{A}_1 \otimes \mathcal{T}^T \right)=9= mM$, Theorem \ref{GenMat_dim} implies $\mathrm{dim}\left(\mathcal{Q}\right)=2\mathrm{dim}\left(\mathcal{C}_1\right)+\mathrm{dim}\left(\mathcal{C}_3\right)=4$. We found that $d\left(\mathcal{Q}\right)=13$. Hence, $\mathcal{Q}$ has the best-known parameters $[18,4,13]_9$ according to \cite{codetables}, even though $\mathcal{C}_1$ and $\mathcal{C}_2$ do not.
\end{example}

\section{QT codes have a GMP structure}
\label{Sec4}
In this section, we show how GMP codes combine MP and QT codes into a single class. This result will be demonstrated in Theorem \ref{QTareGMP}, in which we prove that any QT code can be expressed in a GMP form. Hence, we conclude that QT codes form a subclass of GMP codes. The following result is necessary to prove the main theorem.
\begin{lemma}
\label{lemma1}
Let $\mathcal{Q}=\sum_{k=0}^{r} \mathcal{T}^k \left[ \mathcal{C}_1 \ \cdots  \   \mathcal{C}_M\right] \mathcal{A}_k$ be a GMP code. If $\mathcal{C}_1, \mathcal{C}_2, \ldots,  \mathcal{C}_M$ are $\mathcal{T}$-invariant, then $\mathcal{Q}$ is also $\mathcal{T}$-invariant, i.e., $\mathcal{T}\mathbf{q}\in\mathcal{Q}$ for every $\mathbf{q}\in\mathcal{Q}$. 
\end{lemma}
\begin{proof}
For every $m\times N$ matrix $\mathbf{q}\in\mathcal{Q}$, there exist codewords $\mathbf{c}_i\in\mathcal{C}_i$ such that
\begin{equation*}\begin{split}
\mathcal{T}\mathbf{q}=\mathcal{T}\sum_{k=0}^{r} \mathcal{T}^k \left[ \mathbf{c}_1 \ \cdots  \  \mathbf{c}_M\right] \mathcal{A}_k &=\sum_{k=0}^{r} \mathcal{T}^k \mathcal{T}\left[ \mathbf{c}_1 \ \cdots  \  \mathbf{c}_M\right] \mathcal{A}_k 
=\sum_{k=0}^{r} \mathcal{T}^k \left[ \mathcal{T}\!\mathbf{c}_1 \ \cdots  \  \mathcal{T}\!\mathbf{c}_M\right] \mathcal{A}_k .
\end{split}\end{equation*}
Since $\mathcal{C}_i$ is $\mathcal{T}$-invariant, $\mathcal{T}\mathbf{c}_i \in \mathcal{C}_i$ and, hence, $\mathcal{T}\mathbf{q}\in\mathcal{Q}$.
\end{proof}

From now on and till the end of this section, $\mathcal{Q}$ refers to a $\lambda$-QT code over $\mathbb{F}_q$ of index $\ell$ and co-index $m$. Equivalently, $\mathcal{Q}$ is a $\mathcal{T}_\lambda$-invariant subspace of $\mathbb{F}_q^{m\times \ell}$, where $\mathcal{T}_\lambda$ is the matrix given by \eqref{LambdaTransformation}. We are now ready to prove that any QT code has a GMP structure.

\begin{theorem}
\label{QTareGMP}
Let $\mathcal{T}_\lambda$ be the $m\times m$ matrix given by \eqref{LambdaTransformation}. A code $\mathcal{Q}$ is $\lambda$-QT of index $\ell$ and co-index $m$ if and only if $\mathcal{Q}=\sum_{k=0}^{m-1} \mathcal{T}_\lambda^k \left[ \mathcal{C}_1 \ \cdots  \   \mathcal{C}_M\right] \mathcal{A}_k$ for some $\lambda$-constacyclic codes $\mathcal{C}_1, \ldots, \mathcal{C}_M$ of length $m$ and some $M\times \ell$ matrices $\mathcal{A}_0, \mathcal{A}_1, \ldots, \mathcal{A}_{m-1}$.
\end{theorem}
\begin{proof}
Assume that $\mathcal{Q}=\sum_{k=0}^{m-1} \mathcal{T}_\lambda^k \left[ \mathcal{C}_1 \ \cdots  \   \mathcal{C}_M\right] \mathcal{A}_k$ and suppose that $\mathcal{C}_1, \ldots, \mathcal{C}_M$ are $\lambda$-constacyclic. By Lemma \ref{lemma1},  $\mathcal{Q}$ is $\mathcal{T}_\lambda$-invariant, i.e., $\lambda$-QT. 

Conversely, assume that $\mathcal{Q}$ is $\lambda$-QT with an $M\times \ell$ generator polynomial matrix $\mathfrak{G}$. Let $g_i\left(x\right)$ be the greatest common divisor of all entries of the $i$-th row of $\mathfrak{G}$ for $1\le i\le M$. We set $\mathcal{C}_i=\langle g_i\left(x\right)\rangle$, the $\lambda$-constacyclic code of length $m$ generated by $g_i\left(x\right)$. Then $\mathfrak{G}$ can be written as
$$\mathfrak{G}=\mathrm{diag}\left[g_1\left(x\right), g_2\left(x\right), \ldots, g_M\left(x\right)\right]\mathfrak{G}',$$
where $\mathrm{Row}_i\left(\mathfrak{G}\right)=g_i\left(x\right) \mathrm{Row}_i\left(\mathfrak{G}'\right)$. In addition, we build the $M\times \ell$ matrices $\mathcal{A}_0, \mathcal{A}_1, \ldots, \mathcal{A}_{m-1}$ such that $\mathfrak{G}'=\sum_{k=0}^{m-1}x^k \mathcal{A}_k$. Now we show that $\mathcal{Q}= \sum_{k=0}^{m-1}\mathcal{T}_\lambda^k \left[\mathcal{C}_1 \ \cdots \ \mathcal{C}_M\right] \mathcal{A}_k$. Recall from \eqref{MapPhi} that $\varphi\left(\mathbf{q}\right)$ is the polynomial representation of any codeword $\mathbf{q}\in\mathcal{Q}$. Since $\mathfrak{G}$ generates $\mathcal{Q}$ as a submodule of $\mathcal{R}^\ell$, there exist $a_1(x), a_2(x), \ldots, a_M(x)\in\mathcal{R}$ such that $\varphi\left(\mathbf{q}\right) = \left[a_1(x)\  a_2(x)\  \cdots \  a_M(x)\right]\mathfrak{G}$. There exist codewords $\mathbf{c}_i\in\mathcal{C}_i$ such that $\varphi\left(\mathbf{c}_i\right)=a_i(x)g_i(x)$ for $1\le i\le M$. From \eqref{PhiProperty} we conclude that
\begin{equation*}
\begin{split}
\varphi\left(\mathbf{q}\right)&=\left[a_1(x) \  \cdots \  a_M(x) \right]\mathfrak{G}=\left[a_1(x) \  \cdots \  a_M(x) \right]\mathrm{diag}\left[g_1(x) , \ldots, g_M(x) \right]\mathfrak{G}'\\
&=\left[a_1(x)  g_1(x)  \  \cdots \  a_M(x)  g_M(x)  \right] \sum_{k=0}^{m-1}x^k \mathcal{A}_k 
=\sum_{k=0}^{m-1}x^k \left[\varphi\left(\mathbf{c}_1\right) \  \cdots \  \varphi\left(\mathbf{c}_M\right) \right]  \mathcal{A}_k \\
&=\sum_{k=0}^{m-1}x^k \varphi\left(\left[\mathbf{c}_1 \  \cdots \  \mathbf{c}_M \right]\right)  \mathcal{A}_k 
=\sum_{k=0}^{m-1} \varphi\left(\mathcal{T}_\lambda^k \left[\mathbf{c}_1 \  \cdots \   \mathbf{c}_M \right] \mathcal{A}_k \right)\\
&=\varphi\left(\sum_{k=0}^{m-1} \mathcal{T}_\lambda^k \left[\mathbf{c}_1 \  \mathbf{c}_2 \  \cdots \   \mathbf{c}_M \right] \mathcal{A}_k \right). 
\end{split}
\end{equation*}
Thus $\mathbf{q}=\sum_{k=0}^{m-1} \mathcal{T}_\lambda^k \left[\mathbf{c}_1 \  \cdots \   \mathbf{c}_M \right] \mathcal{A}_k \in \sum_{k=0}^{m-1}\mathcal{T}_\lambda^k \left[\mathcal{C}_1 \ \cdots \ \mathcal{C}_M\right] \mathcal{A}_k$. The other inclusion $\sum_{k=0}^{m-1}\mathcal{T}_\lambda^k \left[\mathcal{C}_1 \ \cdots \ \mathcal{C}_M\right] \mathcal{A}_k \subseteq \mathcal{Q}$ is obtained by reversing the previous equations.
\end{proof}

Theorem \ref{QTareGMP} not only proves that each QT code has a GMP form, but also provides a constructive approach to determine the matrices $\mathcal{A}_0, \mathcal{A}_1, \ldots, \mathcal{A}_{m-1}$ and the codes $\mathcal{C}_1, \mathcal{C}_2, \ldots, \mathcal{C}_M$ from any generator polynomial matrix. The approach will be employed in Examples \ref{Ex1QTtoGMP} and \ref{Ex2QTtoGMP}, and is summarized as follows:
\begin{enumerate}
\item Determine $\mathfrak{G}'$ such that $\mathfrak{G}=\mathrm{diag}\left[g_1(x), \ldots, g_M(x)\right]\mathfrak{G}'$, where $g_i(x)$ is the greatest common divisor of all entries of the $i$-th row of $\mathfrak{G}$ for $1\le i\le M$.
\item Define $\mathcal{C}_i$ to be the $\lambda$-constacyclic code generated by $g_i\left( x\right)$ for $1\le i \le M$.
\item Determine $\mathcal{A}_0, \ldots, \mathcal{A}_{m-1}$ such that $\mathfrak{G}'=\sum_{k=0}^{m-1}x^k \mathcal{A}_k$. 
\end{enumerate}

\begin{example}
\label{Ex1QTtoGMP}
Consider the binary GMP code $$\mathcal{Q}= \left[ \mathcal{C}_1 \ \mathcal{C}_2\right] \mathcal{A}_0+ \mathcal{T} \left[ \mathcal{C}_1 \  \mathcal{C}_2\right] \mathcal{A}_1 + \mathcal{T}^2 \left[ \mathcal{C}_1 \  \mathcal{C}_2\right] \mathcal{A}_2$$ of length $49$, where
\begin{equation*}
\mathcal{A}_0=\begin{bmatrix}
1 & 1 & 1 &1 &1 &1& 1\\
1 & 1 & 1 &0 &1 &0 &0
\end{bmatrix},
\quad
\mathcal{A}_1=\begin{bmatrix}
0 & 0 & 0 &0 &0 & 0 & 0\\
0 & 1 & 1 &1 &0 & 1 &0
\end{bmatrix},
\quad
\mathcal{A}_2=\begin{bmatrix}
0 & 0 & 0 &0 &0 & 0 & 0\\
1 & 1 & 0 &1 &0 & 0 & 1
\end{bmatrix},
\end{equation*}
$\mathcal{T}=\mathcal{T}_\lambda$ with $\lambda=1$, and $\mathcal{C}_1$ and $\mathcal{C}_2$ are the binary cyclic codes of length $m=7$ generated by $g_1\left(x \right)=1+x+x^2+x^4$ and $g_2\left(x \right)=1+x^2+x^3+x^4$ respectively. From Theorem \ref{QTareGMP}, we see that $\mathcal{Q}$ is quasi-cyclic of index $\ell=7$ and co-index $m=7$. But, by Theorem \ref{GenMat_dim}, $\mathrm{dim}\left(\mathcal{Q}\right)= \mathrm{dim}\left(\mathcal{C}_1\right)+ \mathrm{dim}\left(\mathcal{C}_2\right)=6$. We found that $d\left(\mathcal{Q}\right)=24$. Hence, $\mathcal{Q}$ is quasi-cyclic with the best-known parameters $[49,6,24]_2$ according to \cite{codetables}. Furthermore, Theorem \ref{QTareGMP} describes a generator polynomial matrix $\mathfrak{G}$ for $\mathcal{Q}$. Namely, by setting $\mathcal{A}_3=\cdots=\mathcal{A}_6=\mathbf{0}_{2\times 7}$, we get
\begin{equation*}\begin{split}
\mathfrak{G}&=\mathrm{diag}\left[g_1\left(x\right), g_2\left(x\right)\right] \mathfrak{G}'=\mathrm{diag}\left[g_1\left(x\right), g_2\left(x\right)\right] \left(\mathcal{A}_0+x \mathcal{A}_1+x^2 \mathcal{A}_2\right)\\
&=\mathrm{diag}\left[1+x+x^2+x^4, 1+x^2+x^3+x^4\right]
\begin{bmatrix}
1&1&1&1&1&1&1\\
1+x^2&1+x+x^2&1+x&x+x^2&1&x&x^2
\end{bmatrix}.
\end{split}\end{equation*}
\end{example}

\begin{example}
\label{Ex2QTtoGMP}
Let $\mathcal{Q}$ be the $2$-QT code over $\mathbb{F}_5$ of index $\ell=6$ and co-index $m=3$ generated by
\begin{equation*}
\mathfrak{G}=
\begin{bmatrix}
4+2x&3+3x&2+2x& 1+ x& 3+ x& 4+ x \\
2x+x^2&      0 & 4+4x+x^2 &4+2x&3+4x & 2+3x+x^2\\
3+x+2x^2& 1+2x+4x^2 & 0 & 1+2x+4x^2 & 3+x+2x^2  & 3+x+2x^2 
\end{bmatrix}.
\end{equation*}
According to \cite{codetables}, $\mathcal{Q}$ has the best-known parameters $[18,6,10]_5$. We use Theorem \ref{QTareGMP} to describe $\mathcal{Q}$ as a GMP code. That is, $\mathcal{Q}=\sum_{k=0}^{2} \mathcal{T}_2^k \left[ \mathcal{C}_1 \  \mathcal{C}_2 \ \mathcal{C}_3\right] \mathcal{A}_k$ for some $2$-constacyclic codes $\mathcal{C}_1, \mathcal{C}_2, \mathcal{C}_3$ over $\mathbb{F}_5$ of length $3$ and some $3\times 6$ matrices $\mathcal{A}_0, \mathcal{A}_1, \mathcal{A}_2$. We may rewrite $\mathfrak{G}=\mathrm{diag}\left[g_1\left(x\right), g_2\left(x\right),  g_3\left(x\right)\right]\mathfrak{G}'$, where $g_1\left(x\right)=1$, $g_2\left(x\right)=2+x$, $g_3\left(x\right)=4+3x+x^2$, and
\begin{equation*}\begin{split}
\mathfrak{G}'=
\begin{bmatrix}
4+2x &3+3x &2+2x&1+x& 3+x& 4+x\\
x&0&2+x &2&4&1+x\\
2 & 4 & 0 & 4 &2 & 2 
\end{bmatrix}&=\begin{bmatrix}
4&3&2&1&3&4\\
0&0&2&2&4&1\\
2 & 4 & 0 & 4 &2 & 2
\end{bmatrix}+x\begin{bmatrix}
2&3&2&1&1&1\\
1&0&1&0&0&1\\
0&0&0&0&0&0
\end{bmatrix}\\
&=\mathcal{A}_0+x\mathcal{A}_1.
\end{split}\end{equation*}
Define the $2$-constacyclic codes $\mathcal{C}_1=\langle 1\rangle=\mathbb{F}_5^3$, $\mathcal{C}_2=\langle 2+x\rangle$, and $\mathcal{C}_3=\langle 4+3x+x^2\rangle$ over $\mathbb{F}_5$ of length $3$. Therefore, $\mathcal{Q}$ has the GMP structure $\left[ \mathcal{C}_1 \ \mathcal{C}_2 \ \mathcal{C}_3 \right] \mathcal{A}_0+ \mathcal{T}_2 \left[ \mathcal{C}_1 \  \mathcal{C}_2 \ \mathcal{C}_3 \right] \mathcal{A}_1$.
\end{example}

We conclude this section by emphasizing that a generator polynomial matrix of a QT code is not always unique. Consequently, the same QT code can be expressed in multiple GMP structures. For instance, the following matrix is another generator polynomial matrix for the $2$-QT code $\mathcal{Q}$ introduced in Example \ref{Ex2QTtoGMP}. 
\begin{equation*}
\mathfrak{G}_1=
\begin{bmatrix}
1&0&3x+2x^2  &3+x+3x^2 & 4+3x+3x^2  & 3+4x+x^2 \\
0&1& 3+2x&3+4x+3x^2 &2+4x^2 &3x+3x^2 
\end{bmatrix}.
\end{equation*}
In fact, $\mathfrak{G}_1$ is the Hermite normal form generator polynomial matrix. By applying Theorem \ref{QTareGMP} to $\mathfrak{G}_1$ in a manner identical to Example \ref{Ex2QTtoGMP}, we obtain the following GMP structure for $\mathcal{Q}$.
$$\mathcal{Q}= \left[ \mathcal{C}_1 \ \mathcal{C}_2 \right] \mathcal{A}_0+ \mathcal{T}_2 \left[ \mathcal{C}_1 \  \mathcal{C}_2 \right] \mathcal{A}_1+ \mathcal{T}_2^2 \left[ \mathcal{C}_1 \  \mathcal{C}_2 \right] \mathcal{A}_2,$$
where $\mathcal{C}_1=\mathcal{C}_2=\mathbb{F}_5^3$, 
\begin{equation*}
\mathcal{A}_0=
\begin{bmatrix}
1&0&0 &3&4 &3\\
0&1&3&3&2&0
\end{bmatrix}, 
\mathcal{A}_1=\begin{bmatrix}
0&0&3&1&3 &4\\
0&0&2&4&0&3
\end{bmatrix}, \text{ and }
\mathcal{A}_2=\begin{bmatrix}
0&0&2&3&3&1\\
0&0&0&3&4&3
\end{bmatrix}.
\end{equation*}
It is not surprising that there are multiple GMP structures for the same code. This also applies to MP codes; we refer the reader to the discussion about row and column permutations in \cite{Blackmore2001}.

\section{Minimum distance lower bound}
\label{Sec5}
Recall that Equations \eqref{bound1} and \eqref{bound2} provide two lower bounds on the minimum distance of an MP code. In this section, we demonstrate two analogous lower bounds for GMP codes. Accordingly, we adjust the appropriate conditions on $\mathcal{A}_k$ necessary to apply these bounds. These bounds, which generalize \eqref{bound2} and \eqref{bound1}, respectively, are found in Theorems \ref{first_bound} and \ref{second_bound}. The condition in Theorem \ref{second_bound} is less restrictive than that of requiring NSC matrix as in the corresponding bound \eqref{bound1}. We will adopt the following notational convention throughout this section. The minimum distance of $\mathcal{C}_i$ and $\mathcal{Q}$ will be denoted by $d\left(\mathcal{C}_i\right)$ and $d\left(\mathcal{Q}\right)$ respectively. For any $M\times N$ matrix $\mathcal{A}$ and any $1\le t\le M$, the submatrix of $\mathcal{A}$ consisting of the first $t$ rows is denoted by $\mathcal{A}^{(t)}$. Additionally, we define $D_t$ as the minimum distance of the linear code over $\mathbb{F}_q$ of length $N$ spanned by the rows of the matrices $\mathcal{A}_0^{(t)}, \ldots,\mathcal{A}_r^{(t)}$. Let $j_1, j_2, \ldots, j_\tau$ be a sequence of integers such that $1\le j_1 <j_2 <\cdots <j_\tau \le N$. By $\mathcal{A}\left( j_1, j_2, \ldots, j_\tau\right)$ we mean the submatrix of $\mathcal{A}$ which consists of the columns $ j_1, j_2, \ldots, j_\tau$. Formula \eqref{bound2} for the lower bound on the minimum distance of MP codes is generalized in the following result to include GMP codes.
\begin{theorem}
\label{first_bound}
Let $\mathcal{Q}=\sum_{k=0}^{r} \mathcal{T}^k \left[ \mathcal{C}_1 \ \cdots  \   \mathcal{C}_M\right] \mathcal{A}_k$ be a GMP code. Assume that
\begin{equation}
\label{RankLB1}
\mathrm{rank}\left(\begin{bmatrix}
\mathcal{A}_0\\
\mathcal{A}_1\\
\vdots\\
\mathcal{A}_r
\end{bmatrix}\right)=M+\mathrm{rank}\left(\begin{bmatrix}
\mathcal{A}_1\\
\vdots\\
\mathcal{A}_r
\end{bmatrix}\right).
\end{equation}
Then
\begin{equation*}
d\left(\mathcal{Q}\right) \ge \min\left\{D_1 d\left(\mathcal{C}_1\right), D_2 d\left(\mathcal{C}_2\right), \ldots, D_M d\left(\mathcal{C}_M\right) \right\}.
\end{equation*}
\end{theorem}
\begin{proof}
Consider a codeword $\mathbf{q}=\sum_{k=0}^{r} \mathcal{T}^k \left[ \mathbf{c}_1 \ \cdots  \  \mathbf{c}_M\right] \mathcal{A}_k$ of $\mathcal{Q}$. Choose $1\le t\le M$ as the largest integer such that $\mathbf{c}_t \ne \mathbf{0}_{m \times 1}$. Then
\begin{equation*}
\mathbf{q}=\sum_{k=0}^{r} \mathcal{T}^k \left[ \mathbf{c}_1 \ \cdots  \  \mathbf{c}_t\right] \mathcal{A}_k^{(t)}.
\end{equation*}
Let $1\le h \le m$ be chosen arbitrarily such that $\mathrm{Ent}_h \left(\mathbf{c}_t\right) \ne 0$. The $h$-th row of $\mathbf{q}$ is
\begin{equation}
\label{Row_h}
\begin{split}
\mathrm{Row}_h \left(\mathbf{q}\right)&=\sum_{k=0}^{r} \mathrm{Row}_h \left(\mathcal{T}^k \left[ \mathbf{c}_1 \ \cdots  \  \mathbf{c}_t\right]\right) \mathcal{A}_k^{(t)}\\
&=\mathrm{Row}_h \left( \left[ \mathbf{c}_1 \ \cdots  \  \mathbf{c}_t\right]\right) \mathcal{A}_0^{(t)}+\sum_{k=1}^{r} \mathrm{Row}_h \left(\mathcal{T}^k \left[ \mathbf{c}_1 \ \cdots  \  \mathbf{c}_t\right]\right) \mathcal{A}_k^{(t)}\\
&=\left[ \mathrm{Ent}_h \left(\mathbf{c}_1\right) \ \cdots  \  \mathrm{Ent}_h \left(\mathbf{c}_t \right)\right] \mathcal{A}_0^{(t)}+\sum_{k=1}^{r} \mathrm{Row}_h \left(\mathcal{T}^k \left[ \mathbf{c}_1 \ \cdots  \  \mathbf{c}_t\right]\right) \mathcal{A}_k^{(t)}.
\end{split}
\end{equation}
Furthermore, it follows from $\mathrm{Ent}_h \left(\mathbf{c}_t\right) \ne 0$ that $\left[ \mathrm{Ent}_h \left(\mathbf{c}_1\right) \ \cdots  \  \mathrm{Ent}_h \left(\mathbf{c}_t \right)\right] \mathcal{A}_0^{(t)} \ne \mathbf{0}$ since, otherwise, we would have $\mathrm{rank}\left( \mathcal{A}_0^{(t)}\right)<t$, and then $\mathrm{rank}\left( \mathcal{A}_0\right)<M$; however, based on \eqref{RankLB1}, this yields the contradiction
\begin{equation*}\begin{split}
M+\mathrm{rank}\left(\begin{bmatrix}
\mathcal{A}_1\\
\vdots\\
\mathcal{A}_r
\end{bmatrix}\right)&=\mathrm{rank}\left(\begin{bmatrix}
\mathcal{A}_0\\
\mathcal{A}_1\\
\vdots\\
\mathcal{A}_r
\end{bmatrix}\right)\\
&\le \mathrm{rank}\left(\mathcal{A}_0\right)+\mathrm{rank}\left(\begin{bmatrix}
\mathcal{A}_1\\
\vdots\\
\mathcal{A}_r
\end{bmatrix}\right)< M+\mathrm{rank}\left(\begin{bmatrix}
\mathcal{A}_1\\
\vdots\\
\mathcal{A}_r
\end{bmatrix}\right).
\end{split}\end{equation*}

Equation \eqref{Row_h} means that $\mathrm{Row}_h \left(\mathbf{q}\right)$ is a codeword in the linear code of length $N$ that is spanned by the rows of $\mathcal{A}_0^{(t)}, \ldots,\mathcal{A}_r^{(t)}$; this code has minimum distance $D_t$. Then $\mathrm{Row}_h \left(\mathbf{q}\right)$ has a weight of at least $D_t$ if we show that it is nonzero. Assume the contrary, that is,
\begin{equation*}
\left[ \mathrm{Ent}_h \left(\mathbf{c}_1\right) \ \cdots  \  \mathrm{Ent}_h \left(\mathbf{c}_t \right)\right] \mathcal{A}_0^{(t)}=- \sum_{k=1}^{r} \mathrm{Row}_h \left(\mathcal{T}^k \left[ \mathbf{c}_1 \ \cdots  \  \mathbf{c}_t\right]\right) \mathcal{A}_k^{(t)}\ne \mathbf{0}.
\end{equation*}
Then 
\begin{equation*}
\mathrm{rank}\left(\begin{bmatrix}
\mathcal{A}_0\\
\mathcal{A}_1\\
\vdots\\
\mathcal{A}_r
\end{bmatrix}\right)< \mathrm{rank}\left(\mathcal{A}_0\right)+\mathrm{rank}\left(\begin{bmatrix}
\mathcal{A}_1\\
\vdots\\
\mathcal{A}_r
\end{bmatrix}\right)\le M+\mathrm{rank}\left(\begin{bmatrix}
\mathcal{A}_1\\
\vdots\\
\mathcal{A}_r
\end{bmatrix}\right),
\end{equation*}
which once again contradicts \eqref{RankLB1}. We get to the conclusion that every nonzero entry of $\mathbf{c}_t$ has a corresponding row in $\mathbf{q}$ with a weight of at least $D_t$. But $\mathbf{c}_t$ contains at least $d\left(\mathcal{C}_t\right)$ nonzero entries, then the weight of $\mathbf{q}$ is at least $D_t d\left(\mathcal{C}_t\right)$. 
\end{proof}

Condition \eqref{RankLB1} of Theorem \ref{first_bound} in the special case of MP codes, e.g., $\mathcal{A}=\mathcal{A}_0$ and $\mathcal{A}_1, \ldots, \mathcal{A}_r=\mathbf{0}_{M\times N}$, is reduced to $\mathrm{rank}\left(\mathcal{A}\right)=M$. This shows how Theorem \ref{first_bound} expands the bound \eqref{bound2} to GMP codes. It is worth pointing out that the bound of Theorem \ref{first_bound} is tight, meaning that some GMP codes can achieve it. The following examples aim for showing this. 

\begin{example}
\label{Ex_bound1}
The $2\times 5$ matrices 
\begin{equation*}
\mathcal{A}_0=\begin{bmatrix}
0&1&1&1&1\\
1&0&1&1&1
\end{bmatrix}
\quad\text{and}\quad
\mathcal{A}_1=\begin{bmatrix}
1&1&0&1&1\\
1&0&0&1&0
\end{bmatrix}
\end{equation*}
over $\mathbb{F}_2$ satisfy Condition \eqref{RankLB1} of Theorem \ref{first_bound} because
\begin{equation*}
\mathrm{rank}\left(\begin{bmatrix}
\mathcal{A}_0\\
\mathcal{A}_1
\end{bmatrix}\right)=2+\mathrm{rank}\left(\mathcal{A}_1\right)=4.
\end{equation*}
$D_1$ and $D_2$ are defined as the minimum distance of the binary linear codes of length $5$ generated by the matrices
\begin{equation*}
\begin{bmatrix}
\mathrm{Row}_1\left(\mathcal{A}_0\right)\\
\mathrm{Row}_1\left(\mathcal{A}_1\right)
\end{bmatrix}=
\begin{bmatrix}
0&1&1&1&1\\
1&1&0&1&1
\end{bmatrix}
\quad\text{and}\quad
\begin{bmatrix}
\mathcal{A}_0\\
\mathcal{A}_1
\end{bmatrix}=
\begin{bmatrix}
0&1&1&1&1\\
1&0&1&1&1\\
1&1&0&1&1\\
1&0&0&1&0
\end{bmatrix}
\end{equation*}
respectively. Since $D_1=D_2=2$, Theorem \ref{first_bound} implies that $d\left(\mathcal{Q}\right) \ge 2 \min\left\{d\left(\mathcal{C}_1\right), d\left(\mathcal{C}_2\right) \right\}$ for any GMP code $\mathcal{Q}= \left[ \mathcal{C}_1 \ \mathcal{C}_2\right] \mathcal{A}_0 + \mathcal{T} \left[ \mathcal{C}_1\  \mathcal{C}_2\right] \mathcal{A}_1$. Now choose
\begin{equation*}
\mathbf{G}_1=
\begin{bmatrix}
1&0&0&1\\
0&1&0&1\\
0&0&1&1
\end{bmatrix}
\quad \text{and}\quad
\mathcal{T}=
\begin{bmatrix}
 0&1&0&0\\
 0&1&1&0\\
 1&1&0&0\\
 1&1&1&1
\end{bmatrix},
\end{equation*}
then define $\mathcal{C}_1=\mathcal{C}_2$ as the $[4,3,2]_2$ code generated by $\mathbf{G}_1$. From Theorem \ref{GenMat}, $\mathrm{vec}\left(\mathcal{Q}\right)$ has the generator matrix
\begin{equation*}\begin{split}
\mathbf{G}&=\mathrm{diag}\left[\mathbf{G}_1,\mathbf{G}_1\right] \left(\mathcal{A}_0 \otimes \mathcal{I}_4 + \mathcal{A}_1 \otimes \mathcal{T}^T \right)
=\begin{bmatrix}
0&0&1&0&1&0&1&1&1&0&0&1&1&0&1&1&1&0&1&1\\
1&1&1&0&1&0&1&1&0&1&0&1&1&0&1&1&1&0&1&1\\
0&1&0&0&0&1&1&1&0&0&1&1&0&1&1&1&0&1&1&1\\
1&0&1&1&0&0&0&0&1&0&0&1&1&0&1&1&1&0&0&1\\
1&0&1&1&0&0&0&0&0&1&0&1&1&0&1&1&0&1&0&1\\
0&1&1&1&0&0&0&0&0&0&1&1&0&1&1&1&0&0&1&1
\end{bmatrix}.
\end{split}\end{equation*}
The minimum distance of $\mathcal{Q}$ is $d\left(\mathcal{Q}\right)=4$, which coincides with the lower bound of Theorem \ref{first_bound}.
\end{example}

\begin{example}
\label{Ex_bound2}
The $3\times 5$ matrices 
\begin{equation*}
\mathcal{A}_0=\begin{bmatrix}
0&0&2&2&0\\
2&2&0&2&1\\
1&0&2&0&2
\end{bmatrix}
\quad\text{and}\quad
\mathcal{A}_1=\begin{bmatrix}
1&0&2&2&2\\
1&2&2&0&0\\
2&2&1&2&2
\end{bmatrix}
\end{equation*}
over $\mathbb{F}_3$ satisfy Condition \eqref{RankLB1} of Theorem \ref{first_bound} because
\begin{equation*}
\mathrm{rank}\left(\begin{bmatrix}
\mathcal{A}_0\\
\mathcal{A}_1
\end{bmatrix}\right)=3+\mathrm{rank}\left(\mathcal{A}_1\right)=5.
\end{equation*}
By computing the minimum distance of the linear codes over $\mathbb{F}_3$ of length $5$ generated by
\begin{equation*}
\begin{bmatrix}
\mathrm{Row}_1\left(\mathcal{A}_0\right)\\
\mathrm{Row}_1\left(\mathcal{A}_1\right)
\end{bmatrix}=
\begin{bmatrix}
0&0&2&2&0\\
1&0&2&2&2\\
\end{bmatrix}
\  , \ 
\begin{bmatrix}
\mathcal{A}_0^{(2)}\\
\mathcal{A}_1^{(2)}
\end{bmatrix}=
\begin{bmatrix}
0&0&2&2&0\\
2&2&0&2&1\\
1&0&2&2&2\\
1&2&2&0&0
\end{bmatrix}
\ , \ 
\begin{bmatrix}
\mathcal{A}_0\\
\mathcal{A}_1
\end{bmatrix}=
\begin{bmatrix}
0&0&2&2&0\\
2&2&0&2&1\\
1&0&2&0&2\\
1&0&2&2&2\\
1&2&2&0&0\\
2&2&1&2&2
\end{bmatrix},
\end{equation*}
we get $D_1=2$, $D_2=2$, and $D_3=1$, respectively. Theorem \ref{first_bound} implies that $d\left(\mathcal{Q}\right) \ge  \min\left\{2 d\left(\mathcal{C}_1\right), 2 d\left(\mathcal{C}_2\right), d\left(\mathcal{C}_3\right) \right\}$ for any GMP code $\mathcal{Q}= \left[ \mathcal{C}_1 \ \mathcal{C}_2 \ \mathcal{C}_3\right] \mathcal{A}_0 + \mathcal{T} \left[ \mathcal{C}_1 \ \mathcal{C}_2 \ \mathcal{C}_3\right] \mathcal{A}_1$. Now choose
\begin{equation*}
\mathbf{G}_1=
\begin{bmatrix}
1&0&1&0&2\\
0&1&0&1&2
\end{bmatrix}
\quad\text{and}\quad
\mathcal{T}=
\begin{bmatrix}
0&1&2&2&1\\
2&2&2&1&2\\
1&1&1&2&1\\
1&2&0&0&0\\
0&2&1&1&2
\end{bmatrix},
\end{equation*}
then define $\mathcal{C}_1=\mathcal{C}_2$ as the $[5,2,3]_3$ code generated by $\mathbf{G}_1$ and $\mathcal{C}_3$ as the repetition code of length $5$ over $\mathbb{F}_3$. We found that $d\left(\mathcal{Q}\right)=5=\min\left\{2d\left(\mathcal{C}_1\right), 2d\left(\mathcal{C}_2\right), d\left(\mathcal{C}_3\right) \right\}$, which indicates that the lower bound of Theorem \ref{first_bound} is tight.
\end{example}

Since the lower bound of Theorem \ref{first_bound} has the drawback of requiring calculation of the minimum distances $D_t$, we overcome this by introducing another lower bound. This bound is the analogous to the bound of MP codes introduced in \cite[Theorem 3.7]{Blackmore2001}, however it requires a less restrictive condition on $\mathcal{A}_0, \ldots, \mathcal{A}_r$ than employing NSC matrices.

\begin{theorem}
\label{second_bound}
Let $\mathcal{Q}=\sum_{k=0}^{r} \mathcal{T}^k \left[ \mathcal{C}_1 \ \cdots  \   \mathcal{C}_M\right] \mathcal{A}_k$ be a GMP code. Assume that
\begin{equation}
\label{RankLB2}
\mathrm{rank}\left(\begin{bmatrix}
\mathcal{A}_0\\
\mathcal{A}_1\\
\vdots\\
\mathcal{A}_r
\end{bmatrix}\right)=M+\mathrm{rank}\left(\begin{bmatrix}
\mathcal{A}_1\\
\vdots\\
\mathcal{A}_r
\end{bmatrix}\right). 
\end{equation}
For every $1\le t\le M$, let $\tau_t$ be the smallest positive integer such that 
\begin{equation}
\label{RankLB3}
\mathrm{rank}\left(\begin{bmatrix}
\mathcal{A}_0^{\left(t\right)}\left( j_1, j_2, \ldots, j_{\tau_t}\right)\\
\mathcal{A}_1^{\left(t\right)}\left( j_1, j_2, \ldots, j_{\tau_t}\right)\\
\vdots\\
\mathcal{A}_r^{\left(t\right)}\left( j_1, j_2, \ldots, j_{\tau_t}\right)
\end{bmatrix}\right)=1+\mathrm{rank}\left(\begin{bmatrix}
\mathcal{A}_0^{\left(t-1\right)}\left( j_1, j_2, \ldots, j_{\tau_t}\right)\\
\mathcal{A}_1^{\left(t\right)}\left( j_1, j_2, \ldots, j_{\tau_t}\right)\\
\vdots\\
\mathcal{A}_r^{\left(t\right)}\left( j_1, j_2, \ldots, j_{\tau_t}\right)
\end{bmatrix}\right)
\end{equation}
for all $1\le j_1 < j_2 < \cdots < j_{\tau_t} \le N$. Then
\begin{equation*}
d\left(\mathcal{Q}\right) \ge \min\left\{\left(N-\tau_1+1\right) d\left(\mathcal{C}_1\right), \left(N-\tau_2+1\right) d\left(\mathcal{C}_2\right), \ldots, \left(N-\tau_M+1\right) d\left(\mathcal{C}_M\right) \right\}.
\end{equation*}
\end{theorem}
\begin{proof}
Starting as in the proof of Theorem \ref{first_bound} and defining $\mathbf{q}$, $t$, and $h$ as they were defined previously, we get  
\begin{equation}
\label{InProof1}
\begin{split}
\mathrm{Row}_h \left(\mathbf{q}\right) = \mathrm{Ent}_h\left(\mathbf{c}_t\right)\mathrm{Row}_t\left(\mathcal{A}_0\right)+\left[ \mathrm{Ent}_h \left(\mathbf{c}_1\right) \ \cdots  \  \mathrm{Ent}_h \left(\mathbf{c}_{t-1} \right)\right] \mathcal{A}_0^{(t-1)}\\
 +\sum_{k=1}^{r} \mathrm{Row}_h \left(\mathcal{T}^k \left[ \mathbf{c}_1 \ \cdots  \  \mathbf{c}_t\right]\right) \mathcal{A}_k^{(t)}.
\end{split}\end{equation}
We prove by contradiction that $\mathrm{Row}_h \left(\mathbf{q}\right)$ cannot contain $\tau_t$ zeros, and therefore its weight is at least $(N-\tau_t+1)$. Assume that $\mathrm{Row}_h \left(\mathbf{q}\right)$ has $\tau_t$ zeros at the coordinates $j_1, j_2, \ldots, j_{\tau_t}$. Equation \eqref{InProof1} shows that
\begin{equation*}\begin{split}
\mathrm{Row}_t\left(\mathcal{A}_0\right)\left( j_1, \ldots, j_{\tau_t}\right)=\frac{-1}{\mathrm{Ent}_h\left(\mathbf{c}_t\right)}\Big( \left[ \mathrm{Ent}_h \left(\mathbf{c}_1\right) \ \cdots  \  \mathrm{Ent}_h \left(\mathbf{c}_{t-1} \right)\right] \mathcal{A}_0^{(t-1)}\left( j_1, \ldots, j_{\tau_t}\right)\\ + \sum_{k=1}^{r} \mathrm{Row}_h \left(\mathcal{T}^k \left[ \mathbf{c}_1 \ \cdots  \  \mathbf{c}_t\right]\right)  \mathcal{A}_k^{(t)}\left( j_1, \ldots, j_{\tau_t}\right)\Big).
\end{split}\end{equation*}
In fact, this contradicts \eqref{RankLB3} because it means that $\mathrm{Row}_t\left(\mathcal{A}_0\right)\left( j_1, \ldots, j_{\tau_t}\right)$ is in the row span of 
\begin{equation*}
\begin{bmatrix}
\mathcal{A}_0^{\left(t-1\right)}\left( j_1, \ldots, j_{\tau_t}\right)\\
\mathcal{A}_1^{\left(t\right)}\left( j_1, \ldots, j_{\tau_t}\right)\\
\vdots\\
\mathcal{A}_r^{\left(t\right)}\left( j_1, \ldots, j_{\tau_t}\right)
\end{bmatrix}.
\end{equation*}
We get to the conclusion that every nonzero entry of $\mathbf{c}_t$ has a corresponding row in $\mathbf{q}$ with a weight of at least $(N-\tau_t+1)$. But $\mathbf{c}_t$ has at least $d\left(\mathcal{C}_t\right)$ nonzero entries, then the weight of $\mathbf{q}$ is at least $(N-\tau_t+1) d\left(\mathcal{C}_t\right)$.
\end{proof}

It is worth noting that Theorem \ref{second_bound} has several advantages. The first is providing a lower bound similar to \eqref{bound1}, but for the broader class of GMP codes. We note that when applying Theorem \ref{second_bound} to an MP code (i.e., $\mathcal{A}_1, \ldots, \mathcal{A}_r=\mathbf{0}$) with an NSC $\mathcal{A}=\mathcal{A}_0$, Equation \eqref{RankLB3} is satisfied for $\tau_t=t$, and hence, its lower bound is exactly that in \eqref{bound1}. The second advantage is that Theorem \ref{second_bound} requires a less restrictive condition than NSC matrices. For example, consider an MP code $\mathcal{Q}=\left[ \mathcal{C}_1 \  \cdots \ \mathcal{C}_4\right] \mathcal{A}$ over $\mathbb{F}_5$ with
\begin{equation*}
\mathcal{A}=\begin{bmatrix}
 1&1&3&0&1&2&3\\
 3&1&0&0&1&1&3\\
 3&1&0&4&0&1&0\\
 1&3&4&2&4&3&4
\end{bmatrix}.
\end{equation*}
Because $\mathcal{A}$ is not NSC, the lower bound suggested in \eqref{bound1} cannot be applied. However, the lower bound of Theorem \ref{second_bound} can be applied, since Equation \eqref{RankLB3} is satisfied by $\tau_1=2$, $\tau_2=5$, $\tau_3=5$, and $\tau_4=4$. From this, we can immediately conclude 
$$d\left(\mathcal{Q}\right) \ge \min\left\{6 d\left(\mathcal{C}_1\right), 3 d\left(\mathcal{C}_2\right), 3 d\left(\mathcal{C}_3\right), 4 d\left(\mathcal{C}_4\right) \right\}.$$
However, to apply Theorem \ref{first_bound} to the aforementioned matrix $\mathcal{A}$, we find $D_1=6$ and $D_2=D_3=D_4=3$, and hence
$$d\left(\mathcal{Q}\right) \ge \min\left\{6 d\left(\mathcal{C}_1\right), 3 d\left(\mathcal{C}_2\right), 3 d\left(\mathcal{C}_3\right), 3 d\left(\mathcal{C}_4\right) \right\}.$$
Two other advantages for the lower bound of Theorem \ref{second_bound} are that it is a tight bound and it outperforms the bound of Theorem \ref{first_bound} in some cases. This is what we will clarify in the following examples.

\begin{example}
\label{ExampleTwo}
The $2\times 9$ matrices 
\begin{equation*}\begin{split}
\mathcal{A}_0&=\begin{bmatrix}
1&0&0&0&1&1&1&1&1\\
0&0&0&0&1&1&0&0&1
\end{bmatrix}
\ \text{,}\ 
\mathcal{A}_1=\begin{bmatrix}
0&1&0&0&1&0&0&0&0\\
1&0&1&1&0&1&1&1&1
\end{bmatrix}\ ,\\
\mathcal{A}_2&=\begin{bmatrix}
1&1&0&0&0&0&0&0&0\\
0&1&1&0&1&0&0&1&1
\end{bmatrix}
\ \text{,}\ 
\mathcal{A}_3=\begin{bmatrix}
0&0&1&0&0&0&0&1&1\\
1&0&1&1&0&1&1&1&1
\end{bmatrix}
\end{split}\end{equation*}
over $\mathbb{F}_2$ satisfy Condition \eqref{RankLB2} of Theorem \ref{second_bound} because
\begin{equation*}
\mathrm{rank}\left(\begin{bmatrix}
\mathcal{A}_0\\
\mathcal{A}_1\\
\mathcal{A}_2\\
\mathcal{A}_3
\end{bmatrix}\right)=2+\mathrm{rank}\left(\begin{bmatrix}
\mathcal{A}_1\\
\mathcal{A}_2\\
\mathcal{A}_3
\end{bmatrix}\right)=6.
\end{equation*}
We aim first to compare the lower bounds of Theorems \ref{first_bound} and \ref{second_bound}. By computing the minimum distance of the binary linear codes of length $9$ generated by 
\begin{equation*}
\begin{bmatrix}
\mathrm{Row}_1\left(\mathcal{A}_0\right)\\
\mathrm{Row}_1\left(\mathcal{A}_1\right)\\
\mathrm{Row}_1\left(\mathcal{A}_2\right)\\
\mathrm{Row}_1\left(\mathcal{A}_3\right)
\end{bmatrix}
\quad\text{ and }\quad 
\begin{bmatrix}
\mathcal{A}_0\\
\mathcal{A}_1\\
\mathcal{A}_2\\
\mathcal{A}_3
\end{bmatrix},
\end{equation*}
we find $D_1=D_2=2$. Theorem \ref{first_bound} implies that $d\left(\mathcal{Q}\right) \ge 2 \min\left\{d\left(\mathcal{C}_1\right), d\left(\mathcal{C}_2\right)\right\}$ for any GMP code $\mathcal{Q}= \left[ \mathcal{C}_1 \ \mathcal{C}_2 \right] \mathcal{A}_0 + \mathcal{T} \left[ \mathcal{C}_1 \ \mathcal{C}_2 \right] \mathcal{A}_1 + \mathcal{T}^2 \left[ \mathcal{C}_1 \ \mathcal{C}_2 \right] \mathcal{A}_2 + \mathcal{T}^3 \left[ \mathcal{C}_1 \ \mathcal{C}_2 \right] \mathcal{A}_3$. On the other hand, the lower bound of Theorem \ref{second_bound} requires evaluating $\tau_1$ and $\tau_2$. We find that $\tau_1=\tau_2=7$ are the least integers that satisfy 
\begin{equation*}
\mathrm{rank}\left(\begin{bmatrix}
\mathrm{Row}_1\left(\mathcal{A}_0\right)\left( j_1, j_2, \ldots, j_{\tau_1}\right)\\
\mathrm{Row}_1\left(\mathcal{A}_1\right)\left( j_1, j_2, \ldots, j_{\tau_1}\right)\\
\mathrm{Row}_1\left(\mathcal{A}_2\right)\left( j_1, j_2, \ldots, j_{\tau_1}\right)\\
\mathrm{Row}_1\left(\mathcal{A}_3\right)\left( j_1, j_2, \ldots, j_{\tau_1}\right)
\end{bmatrix}\right)=1+\mathrm{rank}\left(\begin{bmatrix}
\mathrm{Row}_1\left(\mathcal{A}_1\right)\left( j_1, j_2, \ldots, j_{\tau_1}\right)\\
\mathrm{Row}_1\left(\mathcal{A}_2\right)\left( j_1, j_2, \ldots, j_{\tau_1}\right)\\
\mathrm{Row}_1\left(\mathcal{A}_3\right)\left( j_1, j_2, \ldots, j_{\tau_1}\right)
\end{bmatrix}\right)
\end{equation*}
and
\begin{equation*}
\mathrm{rank}\left(\begin{bmatrix}
\mathcal{A}_0\left( j_1, j_2, \ldots, j_{\tau_2}\right)\\
\mathcal{A}_1\left( j_1, j_2, \ldots, j_{\tau_2}\right)\\
\mathcal{A}_2\left( j_1, j_2, \ldots, j_{\tau_2}\right)\\
\mathcal{A}_3\left( j_1, j_2, \ldots, j_{\tau_2}\right)
\end{bmatrix}\right)=1+\mathrm{rank}\left(\begin{bmatrix}
\mathrm{Row}_1\left(\mathcal{A}_0\right)\left( j_1, j_2, \ldots, j_{\tau_2}\right)\\
\mathcal{A}_1\left( j_1, j_2, \ldots, j_{\tau_2}\right)\\
\mathcal{A}_2\left( j_1, j_2, \ldots, j_{\tau_2}\right)\\
\mathcal{A}_3\left( j_1, j_2, \ldots, j_{\tau_2}\right)
\end{bmatrix}\right)
\end{equation*}
for all different choices $1\le j_1 < j_2 < \cdots < j_{\tau_1} \le 9$ and $1\le j_1 < j_2 < \cdots < j_{\tau_2} \le 9$ respectively. Then Theorem \ref{second_bound} implies that $d\left(\mathcal{Q}\right) \ge 3 \min\left\{d\left(\mathcal{C}_1\right), d\left(\mathcal{C}_2\right)\right\}$. Now, choose $\mathcal{C}_1=\mathcal{C}_2$ as the $[10,5,4]_2$ code generated by
\begin{equation*}
\mathbf{G}_1=
\begin{bmatrix}
1&0&0&1&0&1&0&1&1&1\\
0&1&0&1&0&1&0&1&0&0\\
0&0&1&1&0&0&0&0&1&1\\
0&0&0&0&1&1&0&0&1&1\\
0&0&0&0&0&0&1&1&1&1
\end{bmatrix}
\end{equation*}
and let $\mathcal{T}$ be the backward identity matrix of size $10$, the resulting GMP code $\mathcal{Q}$ has $d\left(\mathcal{Q}\right)=12=3 \min\left\{d\left(\mathcal{C}_1\right), d\left(\mathcal{C}_2\right)\right\}$. Therefore, while the bound of Theorem \ref{second_bound} is tight, the bound of Theorem \ref{first_bound} is not.
\end{example}

\begin{example}
\label{TightNotBinary}
The $3\times 8$ matrices 
\begin{equation*}
\mathcal{A}_0=\begin{bmatrix}
6&6&4&1&5&6&4&4\\
4&2&0&6&2&4&6&2\\
0&3&3&2&6&1&0&0
\end{bmatrix}
\quad\text{and}\quad
\mathcal{A}_1=\begin{bmatrix}
3&6&5&5&4&4&1&3\\
4&5&4&1&3&0&0&5\\
2&4&0&6&4&0&0&2
\end{bmatrix}
\end{equation*}
over $\mathbb{F}_7$ satisfy Condition \eqref{RankLB2} of Theorem \ref{second_bound}. The lower bound of Theorem \ref{first_bound} requires the minimum distances of the linear codes over $\mathbb{F}_7$ of length $8$ generated by 
\begin{equation*}
\begin{bmatrix}
\mathcal{A}_0^{(1)}\\
\mathcal{A}_1^{(1)}
\end{bmatrix}
\ ,\ 
\begin{bmatrix}
\mathcal{A}_0^{(2)}\\
\mathcal{A}_1^{(2)}
\end{bmatrix}
\ , \text{ and } \ 
\begin{bmatrix}
\mathcal{A}_0\\
\mathcal{A}_1
\end{bmatrix},
\end{equation*}
these are $D_1=6$, $D_2=3$, and $D_3=2$, respectively. Theorem \ref{first_bound} implies that $d\left(\mathcal{Q}\right) \ge  \min\left\{6d\left(\mathcal{C}_1\right), 3d\left(\mathcal{C}_2\right), 2d\left(\mathcal{C}_3\right)\right\}$ for any GMP code $\mathcal{Q}= \left[ \mathcal{C}_1 \ \mathcal{C}_2 \ \mathcal{C}_3 \right] \mathcal{A}_0 + \mathcal{T} \left[ \mathcal{C}_1 \ \mathcal{C}_2 \ \mathcal{C}_3 \right] \mathcal{A}_1$. On the other hand, the lower bound of Theorem \ref{second_bound} requires the least positive integers $\tau_1$, $\tau_2$, and $\tau_3$ such that
\begin{equation*}\begin{split}
\mathrm{rank}\left(\begin{bmatrix}
\mathrm{Row}_1\left(\mathcal{A}_0\right)\left( j_1, j_2, \ldots, j_{\tau_1}\right)\\
\mathrm{Row}_1\left(\mathcal{A}_1\right)\left( j_1, j_2, \ldots, j_{\tau_1}\right)
\end{bmatrix}\right)&=1+\mathrm{rank}\left(
\mathrm{Row}_1\left(\mathcal{A}_1\right)\left( j_1, j_2, \ldots, j_{\tau_1}\right)
\right),\\
\mathrm{rank}\left(\begin{bmatrix}
\mathcal{A}_0^{(2)}\left( j_1, j_2, \ldots, j_{\tau_2}\right)\\
\mathcal{A}_1^{(2)}\left( j_1, j_2, \ldots, j_{\tau_2}\right)
\end{bmatrix}\right)&=1+\mathrm{rank}\left(\begin{bmatrix}
\mathrm{Row}_1\left(\mathcal{A}_0\right)\left( j_1, j_2, \ldots, j_{\tau_2}\right)\\
\mathcal{A}_1^{(2)}\left( j_1, j_2, \ldots, j_{\tau_2}\right)
\end{bmatrix}\right),\\
\mathrm{rank}\left(\begin{bmatrix}
\mathcal{A}_0\left( j_1, j_2, \ldots, j_{\tau_3}\right)\\
\mathcal{A}_1\left( j_1, j_2, \ldots, j_{\tau_3}\right)
\end{bmatrix}\right)&=1+\mathrm{rank}\left(\begin{bmatrix}
\mathcal{A}_0^{(2)}\left( j_1, j_2, \ldots, j_{\tau_3}\right)\\
\mathcal{A}_1\left( j_1, j_2, \ldots, j_{\tau_3}\right)
\end{bmatrix}\right),
\end{split}
\end{equation*}
for all different choices of matrix columns. We find that $\tau_1=3$ and $\tau_2=\tau_3=6$. Then Theorem \ref{second_bound} implies that $d\left(\mathcal{Q}\right) \ge  \min\left\{6 d\left(\mathcal{C}_1\right), 3d\left(\mathcal{C}_2\right), 3d\left(\mathcal{C}_3\right)\right\}$. Now, choose 
\begin{equation*}
\mathcal{T}=
\begin{bmatrix}
0&2&0&2\\
0&6&0&3\\
1&3&4&1\\
6&4&0&6
\end{bmatrix}\ ,\ 
\mathbf{G}_2=
\begin{bmatrix}
1&0&6&3\\
0&1&3&3
\end{bmatrix}\  , \ \mathbf{G}_3=
\begin{bmatrix}
1&0&3&5\\
0&1&4&0
\end{bmatrix},
\end{equation*}
and define $\mathcal{C}_1=\mathbb{F}_7^4$, $\mathcal{C}_2$ the $[4,2,3]_7$ code generated by $\mathbf{G}_2$, and $\mathcal{C}_3$ the $[4,2,2]_7$ code generated by $\mathbf{G}_3$. With these settings, we find that the minimum distance of $\mathcal{Q}$ is $d\left(\mathcal{Q}\right)=6$. Once again, the lower bound of Theorem \ref{second_bound} is tight, but that of Theorem \ref{first_bound} is not tight.
\end{example}

\begin{example}
Theorem \ref{second_bound} establishes tight lower bounds on the minimum distances of the GMP codes introduced in Examples \ref{Ex_bound1} and \ref{Ex_bound2}. Specifically, for Example \ref{Ex_bound1}, we have $\tau_1=\tau_2=4$, so $d\left(\mathcal{Q}\right) \ge  \min\left\{2d\left(\mathcal{C}_1\right), 2d\left(\mathcal{C}_2\right) \right\}$; in contrast, $\tau_1=\tau_2=4$ and $\tau_3=5$ for Example \ref{Ex_bound2}, and hence, $d\left(\mathcal{Q}\right) \ge  \min\left\{2d\left(\mathcal{C}_1\right), 2d\left(\mathcal{C}_2\right), d\left(\mathcal{C}_3\right) \right\}$. Therefore, in these examples, the lower bounds of Theorems \ref{first_bound} and \ref{second_bound} coincide and are tight.
\end{example}

\section{Conclusion}
\label{Concl}
We introduced a comprehensive class, the class of GMP codes, which comprises MP and QT codes simultaneously. We presented several examples to illustrate that this class contains many codes with the best-known parameters that are neither MP nor QT. We presented a generator matrix formula for any linear GMP code that generalizes the particular case of classical MP codes. We also determined the GMP code's size and dimension. We proved that GMP codes include all QT codes by showing that any QT code has a GMP structure. Two different lower bounds on the minimum distance of GMP codes were then established. We have shown the tightness of these bounds. These bounds seem to generalize two other bounds that have been mentioned in MP codes literature, however, they require a less restrictive condition than NSC matrices.

\bibliographystyle{unsrt}

\end{document}